\RequirePackage{amsmath}
\documentclass[runningheads]{llncs}
\usepackage{graphicx}

\usepackage{url}
\usepackage{mathtools}
\usepackage{algorithm}
\usepackage{amssymb}
\usepackage{cite}
\usepackage{mdframed}
\usepackage[noend]{algpseudocode}
\algnewcommand\algorithmicforeach{\textbf{for each}}
\algdef{S}[FOR]{ForEach}[1]{\algorithmicforeach\ #1\ \algorithmicdo}

\usepackage{textcomp}
\usepackage{amsmath}
\usepackage{tikz}
\usepackage{mathdots}
\usepackage{yhmath}
\usepackage{cancel}
\usepackage{color}
\usepackage{siunitx}
\usepackage{array}
\usepackage{multirow}
\usepackage{amssymb}
\usepackage{gensymb}
\usepackage{tabularx}
\usepackage{booktabs}
\usetikzlibrary{fadings}
\usetikzlibrary{patterns}
\usetikzlibrary{shadows.blur}

\algnewcommand\algorithmicswitch{\textbf{switch}}
\algnewcommand\algorithmiccase{\textbf{case}}
\algnewcommand\algorithmicassert{\texttt{assert}}
\algnewcommand\Assert[1]{\State \algorithmicassert(#1)}%
\algdef{SE}[SWITCH]{Switch}{EndSwitch}[1]{\algorithmicswitch\ #1\ \algorithmicdo}{\algorithmicend\ \algorithmicswitch}%
\algdef{SE}[CASE]{Case}{EndCase}[1]{\algorithmiccase\ #1}{\algorithmicend\ \algorithmiccase}%
\algtext*{EndSwitch}%
\algtext*{EndCase}%

\makeatletter
\def\BState{\State\hskip-\ALG@thistlm}
\makeatother

\title{Game Description Logic with Integers: \\  
A GDL Numerical Extension
}

\titlerunning{Game Description Logic with Integers}

\author{
Munyque Mittelmann 
\and
Laurent Perrussel 
}
\authorrunning{Mittelmann and Perrussel}

\institute{
  Universit\'e de Toulouse - IRIT,
  Toulouse,  France \\
  \email{\{munyque.mittelmann, laurent.perrussel\}@irit.fr}
}

\begin{document}

\maketitle

\begin{abstract}
Many problems can be viewed as games, where one or more agents try to ensure that certain objectives hold no matter the behavior from the environment and other agents.
In recent years, a number of logical formalisms have been proposed for specifying games among which the Game Description Language (GDL) was established as the official language for General Game Playing. 
Although numbers are recurring in games, the description of games with numerical features in GDL requires the enumeration from all possible numeric values and the relation among them. Thereby, in this paper, we introduce the Game Description Logic with Integers (GDLZ) to describe games with numerical variables, numerical parameters, as well as to 
perform numerical comparisons. 
We compare our approach with GDL and show that when describing the same game, GDLZ is more compact.
\end{abstract}
\section{Introduction}
\label{sec:Introduction}

Many problems, as multiagent planning or process synchronization, can be \linebreak viewed
as games, where one or more agents try to ensure that certain objectives hold no matter the behavior from the environment and other agents \cite{Giacomo2010}. 
Thereby, a number of logical formalisms have been proposed for specifying game structures and its properties, 
 such as the Game Logic \cite{Parikh1985,Pauly2003}, the Dynamic Game Logic for sequential \cite{VanBenthem2001} and simultaneous games \cite{VanBenthem2008}, the GameGolog language \cite{Giacomo2010} and so on. 
Among this formalisms, the Game Description Language (GDL) \cite{Genesereth2005,gdl_specification} 
has been established  
as the official language for the General Game Playing (GGP) Competition. 
Due to the GDL limitations, 
such as its restriction to deterministic games with complete state information, several works investigate GDL extensions to improve its expressiveness. 
Zhang and Thielscher (2014) \cite{Zhang2015Representing} provide a GDL extension using a modality for linear time and state transition structures. They also propose two dual connectives to express preferences in strategies.

Another extension is called GDL with Incomplete Information (GDL-II) and it was proposed to describe nondeterministic games with randomness and incomplete state knowledge \cite{Thielscher2010,Schiffel2014}.
A different approach to deal
with this problem is the Epistemic GDL, that allows to represent imperfect information games and provides a semantical model that can be used for reasoning about game information and players's epistemic status \cite{Jiang2016a}.
GDL with Imperfect Information and Introspection (GDL-III) is an extension of GDL-II to include epistemic games, which are characterized by rules that depend on the knowledge of players  \cite{Thielscher2016, Thielscher2017}. 
In order to model how
agents can cooperate to achieve a desirable goal, 
Jiang et al. (2014) present a framework to combine GDL with the coalition operators from Alternating-time Temporal Logic and prioritized strategy connectives \cite{Jiang2014GDLmeetsATL}.

Although numbers are recurring in game descriptions (e.g. Monopoly, Nim game), neither GDL or its extensions incorporate numerical features. In these approaches, numbers can be designed as index in propositions or actions but not directly used as state variables. Thereby, describing games with numerical features can lead to an exhaustive enumeration of all possible numeric values and the relation between them.
In the context of planning problems, numerical features  have been introduced in Planning Domain Description Language (PDDL) by its first versions \cite{Ghallab98, McDermott2000} and 
improved by PDDL 2.1 \cite{Fox2003, Gerevini2008}.
In PDDL 2.1, a world state contains
an assignment of values to a set of numerical variables. 
These variables can be 
modified by action effects
and used in expressions to describe actions' preconditions and planning goals.

Similarly to the approach of PDDL 2.1, 
in this paper, we introduce the GDL extension Game Description Logic with Integers (GDLZ) that incorporates numerical variables, parameters and comparisons.
Regarding that board games are mainly described with discrete values, our approach only considers the integer set.
We compare our approach with GDL and show that a game description in GDLZ is more compact than the corresponding 
description in GDL.

This paper is organized as follows. In Section \ref{sec:GameDescriptionIntegers}, we introduce the framework 
by means of state transition structures 
and we present
the language syntax and semantics. 
In Section \ref{sec:GDLxZ}, we define the translation between GDLZ and GDL and we compare both languages. 
Section \ref{sec:conclusions} concludes the paper, bringing final considerations.  

\section{Game Description Logic with Integers}
\label{sec:GameDescriptionIntegers}

In this section, we introduce a logical framework for game specification with integer numbers. The framework is an extension from the GDL state transition model and language \cite{Zhang2015Representing}, such that it defines numerical variables and parameters. We call the framework Game Description Logic with Integers, denoted GDLZ.

To describe a game, we first define a game signature, that
specifies who are the players (the agents), what are the possible actions for each player and what are the aspects that describe
each state in the game (the propositions and numerical variables).
We define a game signature as follows:

\begin{definition}
\label{def:signature}
A \textit{game signature} $\mathcal{S}$ is a tuple $(N, \mathcal{A}, \Phi , X)$, where:

\begin{itemize}
\item  $N = \{r_1, r_2, \cdot\cdot\cdot, r_\mathsf{k}\}$ is a nonempty finite set of \textit{agents};
\item $\mathcal{A} = \bigcup_{r \in N} A^r $ where $A^r = \{ a^r_1(\bar{z}_1), \cdot\cdot\cdot, a^r_m(\bar{z}_m)\} $ consists of a nonempty set of \textit{actions} performed by agent $r \in N$, 
where $\bar{z}_i \in \mathbb{Z}^{l}$ 
is a possibly empty tuple of $l$ integer values representing the parameters for the action $a^r_i$, 
$i \leq m$ and $l \in \mathbb{N}$. 
		For convenience, we occasionally write  $a^r_i$ for denoting an action $a^r_i(\bar{z}_i) \in \mathcal{A}$; 
\item  $\Phi = \{p, q, \cdot\cdot\cdot \}$ is a finite set of atomic propositions for specifying individual features of a game state; 
\item $X = \langle x_1, x_2, \cdot\cdot\cdot, x_n\rangle$ is a tuple 
of numerical variables for specifying numerical features of a game state.
\end{itemize}

 
\end{definition}


Given a game signature, we define 
a state transition model, that allows us to represent the key aspects of
a game, such as the winning states for each agent, the legal actions in each state and the transitions between game states.



\begin{definition}
\label{def:model}
Given a game signature $\mathcal{S} = (N, \mathcal{A}, \Phi , X)$, a state transition ST \textit{model} $M$ is a tuple $(W, \bar{w}, T, L, U, g, \pi_\Phi , \pi_\mathbb{Z})$, where:

\begin{itemize}
\item  $W$ is a nonempty set of \textit{states};
\item $\bar{w} \in W$ is the \textit{initial} state;
\item $T \subseteq W$ is a set of \textit{terminal} states;
\item $ L \subseteq W \times \mathcal{A}$ is a \textit{legality} relation, describing the legal actions at each state; 
\item $U : W \times D \rightarrow W$ is an \textit{update} function, where $D = \prod_{r \in N} A^r$ denote the set of joint actions, specifying the transitions for each joint state;
\item $g : N \rightarrow 2^W$ is a \textit{goal} function, specifying the winning states for each agent;
\item $\pi_\Phi : W \to 2^\Phi $ is the valuation function for the state propositions;
\item $\pi_\mathbb{Z}: W \to \mathbb{Z}^n$ is the valuation function for the state numerical variables,  such that $\pi_\mathbb{Z}(w)$ is a tuple of integer values assigned to the variables $X$ at state $w \in W$. %
Let $\pi_\mathbb{Z}^i(w)$ denote the $i$-th value of $\pi_\mathbb{Z}(w)$.
\end{itemize}
\end{definition}

Given $d\in D$, let $d(r)$ be the individual action for agent $r$ in the joint action $d$. Let $L(w) = \{a \in \mathcal{A} \mid (w,a) \in L \}$ be the set of all legal actions at state $w$.

\begin{definition}
\label{def:path}
Given an ST-model $M = (W, \bar{w}, T, L, U, g, \pi_\Phi , \pi_\mathbb{Z})$, a \textit{path} is a finite sequence of states $\bar{w}\stackrel{d_1}{\to} w_1 \stackrel{d_2}{\to}\cdot\cdot\cdot \stackrel{d_e}{\to} w_e$ such that $e\geq 0$ and for any $j \in \{1, \cdot\cdot\cdot, e\}$:  (i) $\{w_0, \cdot\cdot\cdot, w_{e-1}\} \cap T = \emptyset$, where $w_0 = \bar{w}$; (ii) $d_j(r) \in L(w_{j-1})$ for any $r \in N$; and (iii) $w_j = U(w_{j-1}, d_j)$. 

\end{definition}

A path $\delta$ is \textit{complete} if $w_e \in T$.
Given $\delta \in \mathcal{P}$, let $\delta[j]$ denote the $j$-th reachable state of $\delta$, $\theta(\delta, j)$ denotes the joint action taken at stage $j$ of $\delta$; and $\theta_r(\delta, j)$ denotes the action of agent $r$ taken at stage $j$ of $\delta$. 
Finally, the length of a path $\lambda$, written $|\lambda|$, is defined as the number of joint actions.



Describing a game with the ST-model is not practical, especially when modeling large games. Hereby, given a game signature $\mathcal{S} = (N, \mathcal{A}, \Phi, X)$, we introduce a variant of the language for GDL ($\mathcal{L}_{GDL}$ for short) 
to describe a GDLZ game in a more compact way by encoding its rules.

\subsection{Syntax}
\label{sec:GDLZ-syntax}
The \textit{language} is denoted by $\mathcal{L}_{GDLZ}$ and a \textit{formula} $\varphi$ in $\mathcal{L}_{GDLZ}$ is defined by the following Backus-Naur Form (BNF) grammar:

\begin{align*}
\varphi ::= p \mid initial \mid terminal \mid legal(a^r ( \bar{z}))\mid wins(r) \mid does(a ^r(\bar{z}))  \mid \neg \varphi \mid \varphi \land \varphi \mid 
\\ \bigcirc \varphi \mid
 z>z\mid z<z\mid z=z \mid \langle\bar{z}\rangle
\end{align*}

\noindent where, $p \in \Phi , r \in N, a^r\in A^r$,  $\bar{z}$ is a number list and $z$ is a numerical term. 

Let $\varepsilon$ denote the empty word. A number list $\bar{z}$ is defined as:
\[
\bar{z} :: = z \mid  z, \bar{z} \mid \varepsilon.
\]

Finally, a numerical term $z$ is defined by $\mathcal{L}_z$, which is generated by the following BNF: 
\[
z ::= z' \mid x'\mid add(z, z)\mid sub(z,z) \mid min(z,z)\mid max(z,z) 
\]
where $ z' \in \mathbb{Z}$ and $ x' \in X$. 

Other connectives $\lor, \to, \leftrightarrow, \top$ and $ \bot $ are defined by $\neg$ and $\land$ in the standard way.
The comparison operators $\leq$, $\geq$ and $\neq$ are defined by $\lor,  >, < $ and $ =$, respectively, as follows: (i) $z_1 < z_2\lor  z_1 = z_2$,  (ii) $z_1 > z_2\lor  z_1 = z_2$ and (iii) $z_1 > z_2 \lor z_1<z_2$. 

Intuitively, $initial$ and $terminal$ specify the initial state and the terminal state, respectively; $does(a^r(\bar{z}))$ asserts that agent $r$ takes action $a$ with the parameters $\bar{z}$ at the current state; $legal(a^r(\bar{z}))$ asserts that agent $r$ is allowed to take action $a$ with the parameters $\bar{z}$ at the current state; and $wins(r)$ asserts that agent $r$ wins at the current state. The formula $\bigcirc \varphi$ means ``$\varphi$ holds at the next state''. The formulas $z_1 > z_2$, $z_1<z_2$, $z_1=z_2$ means that a numerical term $z_1$ is greater, less and equal to a numerical term $z_2$, respectively.
Finally, $\langle \bar{z}\rangle$ asserts the current values for the numerical variables, i.e. the $i$-th variable in $X$ has the $i$-th value in $\bar{z}$, for $0\leq i \leq |X|$. Notice that $\langle \bar{z}\rangle$ could be represented by a conjunction over each $x_i \in X$ of formulas $x_i = z_i$, where $z_i \in \mathcal{L}_{z}$ is the current value of the variable $x_i$. However, $\langle \bar{z}\rangle$ provides a short cut and it is more meaningful, in the sense that it is strictly related to the valuation of the numerical variables in a given state.

For numerical terms, $add(z_1, z_2)$ and $sub(z_1,z_2)$ specify the value obtained by adding and subtracting $z_2$ from $z_1$, respectively. The formulas $min(z_1,z_2)$ and $max(z_1,z_2)$ specify the minimum and maximum value between $z_1$ and $z_2$, respectively. The extension of the comparison operators $>, <, =$, $\leq$, $\geq$ and $\neq$ to multiple arguments is straightforward.


If $\varphi$ is not in the form $\neg \varphi'$, $\bigcirc \varphi'$ or $ \varphi' \land \varphi''$, for any $\varphi', \varphi'' \in \mathcal{L}_{GLDZ}$, then $\varphi$ is called an atomic formula.
We say that a numerical variable occurs in an atomic formula $\varphi$ if (i) $\varphi$ is either in the form $legal(a^r(\bar{z}))$, $does(a^r(\bar{z}))$ or $\langle \bar{z}\rangle$ and  there is a  $x \in X$ in the numerical list $\bar{z}$; (ii) $\varphi$ is either in the form $z_1 < z_2$, $z_1 > z_2$ or $z_1 = z_2$ and $z_1 \in X$ or $z_2 \in X$.


\subsection{Semantics}
\label{sec:GDLZ-Semantics}

The semantics for the GDLZ language is given in two steps. First, we define function $v$ to assign the meaning of numerical terms $z \in \mathcal{L}_z$ in a specified  state (Definition \ref{def:funcionv}). Next, a formula $\varphi \in \mathcal{L}_{GDLZ}$ is interpreted with respect to a stage in a path (Definition \ref{def:semantics}).

\begin{definition}
\label{def:funcionv}
Given an ST-model $M$, a state $w$ 
and the functions $minimum$ and $maximum$\footnote{Through the rest of this paper, the functions $minimum(a, b)$ and $maximum(a,b)$ respectively return the minimum and maximum value between $a,b \in \mathbb{Z}$.} let us define function $v: W \times \mathcal{L}_z \rightarrow \mathbb{Z}$, associating any $z_i \in \mathcal{L}_z$ in a state $w \in W$ to a number in $\mathbb{Z}$: 

\[v(z_i, w) = \begin{dcases}  z_i & \text{if }  z_i \in \mathbb{Z}
\\ \pi_{\mathbb{Z}}^i(w) & \text{if } z_i = x_i \text{ \& } x_i \in X  
\\v(z_i', w) + v(z_i'', w) &\text{if } z_i = add(z_i', z_i'')
\\v(z_i', w) - v(z_i'', w) &\text{if } z_i = sub(z_i', z_i'')
\\minimum(v(z_i', w),v(z_i'', w)) &\text{if } z_i = min(z_i', z_i'')
\\maximum(v(z_i', w),v(z_i'', w)) &\text{if } z_i = min(z_i', z_i'') 
\end{dcases}\]

\end{definition}

\begin{definition}
\label{def:semantics}
Let $M$ be an ST-Model. Given a complete path $\delta$ of $M$, a stage $j$ on $\delta$,  
a formula $\varphi \in \mathcal{L}_{GDLZ}$ and function $v$, we say $\varphi$ is true (or satisfied) at $j$ of $\delta$ under $M$, denoted by $M, \delta, j \models \varphi$, according with the following definition:

\[ \arraycolsep=1.4pt\def \arraystretch{1.2}
\begin{array}{lllll}
M, \delta, j \models p & \qquad   & \text{iff} & \qquad \qquad  &  p \in \pi_\Phi (\delta [j]) \\

M, \delta, j  \models \neg \varphi  &  &\text{iff}  &&  M, \delta, j \not\models \varphi \\

M, \delta, j \models \varphi_1 \land \varphi_2  &  &\text{iff}  &&  M, \delta, j \models \varphi_1 \text{ and } M, \delta, j \models \varphi_2 \\

M, \delta, j \models initial  &  &\text{iff}  &&  \delta[j] = \bar{w} \\

M, \delta, j  \models terminal  &  &\text{iff}  &&  \delta[j] \in T\\

M, \delta, j \models wins(r)  &  &\text{iff}  && \delta[ j] \in g(r)\\


M, \delta, j  \models legal(a^r(\bar{z}))  &  &\text{iff}  &&  a^r(v(z, \delta [j]): z \in \bar{z}) \in L(\delta[j]) 
\\

M, \delta, j \models  does(a^r(\bar{z}))  &  &\text{iff}  &&  \theta_r (\delta, j) = a^r(v(z, \delta [j]): z \in \bar{z}) 
\\

M, \delta, j \models \bigcirc \varphi  &  &\text{iff}  &&  \text{if } j < |\delta|\text{, then } M, \delta, j+1 \models \varphi \\

M, \delta, j  \models  z_1 > z_2 &  & \text{iff} &  &  v(z_1, \delta [j]) > v(z_2, \delta [j]) \\

M, \delta, j \models  z_1 < z_2 &  & \text{iff} &  &   v(z_1, \delta [j]) 	< v(z_2, \delta [j])\\

M, \delta, j  \models  z_1 =  z_2 &  & \text{iff} &  &   v(z_1,\delta [j]) = v(z_2, \delta [j])\\


M, \delta, j \models \langle\bar{z}\rangle  &  &\text{iff}  && 
\langle v(z, \delta [j]): z \in \bar{z} \rangle = \pi_\mathbb{Z}(\delta[j]) 
\\ 
\end{array}\]
\end{definition}


A formula $\varphi$ is globally true through $\delta$, denoted by $M, \delta \models \varphi $, if $M, \delta, j \models \varphi$ for any stage $j$ of $\delta$. A formula $\varphi$ is \textit{globally true} in an ST-Model $M$, written $M \models \varphi$, if $M, \delta \models \varphi$ for all complete paths $\delta$ in $M$, that is, $\varphi$ is true at every reachable state. 
A formula $\varphi$ is \textit{valid}, denoted by $\models \varphi$, if it is globally true in every ST-model of an appropriate signature. Finally, let $\Sigma$ be a set of formulas in $\mathcal{L}_{GDLZ}$, then $M$ is a \textit{model} of $\Sigma$ if $M \models \varphi$ for all $\varphi \in \Sigma$.

Whenever $j\geq |\delta|$, the validity of $M, \delta, j \models \bigcirc \varphi$ is irrelevant, since 
$\delta[j]$ is the last state reachable in $\delta$. A formula $\langle \bar{z} \rangle$ is valid at a stage $j$ in a path $\delta$ under $M$ only when it corresponds to the valuation of the numerical variables at $\delta[j]$.

The following propositions show that if a player does an action at a stage in a path, then (i) he does not any other action in the same stage and (ii) the action taken is legal.

\begin{proposition}
\label{prop:doesnotdoes}
$\models does(a^r(\bar{z})) \rightarrow $ 
$ \bigwedge_{b^r \neq a^r \in A^r}\bigwedge_{\bar{z}\text{}' \neq \bar{z} \in \mathbb{Z}^n}  \neg does(b^r(\bar{z}\text{}'))$.
\end{proposition}

\begin{proof}
Assume $M, \delta, j\models does(a^r(\bar{z})) $ for some $\bar{z}$ iff $\theta_r(\delta, j) = does(a^r(\bar{z}))$. Then, for any $\bar{z}' \neq \bar{z}, b^r \neq a^r$, $\theta_r(\delta, j) \neq does(b^r$ $(\bar{z}'))$. Thereby, $M, \delta, j \not \models \bigwedge_{b^r\neq a^r \in A^r}\bigwedge_{\bar{z}\text{}' \neq \bar{z} \in \mathbb{Z}^n}  does(b^r(\bar{z}\text{}')) $ and $M, \delta, j \models \bigwedge_{b^r\neq a^r \in A^r}\bigwedge_{\bar{z}\text{}' \neq \bar{z} \in \mathbb{Z}^n}\neg does(b^r$ $(\bar{z}\text{}'))$.
\end{proof}

\begin{proposition}
\label{prop:doeslegal}
$\models does(a^r(\bar{z})) \rightarrow legal(a^r(\bar{z}))$.
\end{proposition}
\begin{proof}
Assume $M, \delta, j  \models does(a^r(\bar{z}))$, then $a^r = \theta^r(\delta, j)$. And by the definition of $\delta$, $a^r(\bar{z}) \in L(\delta[j])$, so $M, \delta, j \models legal(a^r(\bar{z}))$.  
\end{proof}

Next, we illustrate the representation of a game with numerical features in GDLZ. First, we define the game signature and the game description in $\mathcal{L}_{GDLZ}$. Next, we define the ST-model by which it is possible to evaluate the $\mathcal{L}_{GDLZ}$ semantics. Finally, we illustrate a path in the game.

\begin{example}
\label{exemplo}
\textit{\textbf{($\langle \gamma_1, \cdot \cdot \cdot, \gamma_k \rangle$-Nim Game)} A $\langle \gamma_1, \cdot \cdot \cdot, \gamma_k \rangle$-Nim Game consists in $k$ heaps. Each heap starts with $\gamma_i$ sticks, where $1 \leq i \leq k$. Two players take turns in removing sticks from one heap. The game ends when all heaps are empty. A player wins if it is not his turn when the game ends.  }

To represent a $\langle \gamma_1, \cdot \cdot \cdot, \gamma_k \rangle$-Nim Game 
in terms in GDLZ, we first specify the agents, the actions, the propositions, and the numerical variables involved in the game. Thus,  the game signature, written $\mathcal{S}_{k\text{-nim}}$, is described as follows:  

\begin{itemize}
\item $N_{k\text{-nim}} = \{Player_1, Player_2\}$;
\item $A_{k\text{-nim}}^r = \{reduce^r (m,s)\mid s\in \mathbb{N}, 1\leq m\leq k\}\cup \{noop^r\}$, where $reduce^r (m,s)$ denotes the action that player $r$ removes $s$ sticks from the $m$-th heap and $noop^r$ denotes that player $r$ does action noop; 

\item $\Phi _{k\text{-nim}} =\{ turn(r) \mid r \in \{Player_1, Player_2\} \}$, where $turn(r)$ says that it is player $r$'s turn now;
\item $X_{k\text{-nim}} = \langle heap_i \mid 1\leq i \leq k\rangle$, where $heaps_i$ represents the amount of sticks in the $i$-th heap.

\end{itemize}

Given a player $r \in N_{k-min}$, we denote $-r$ as the opponent of $r$, i.e.
 $-r = Player_1$ if $r = Player_2$ and $-r = Player_2$ otherwise. 
The rules of the $\langle \gamma_1, \cdot \cdot \cdot, \gamma_k \rangle$-Nim Game can be expressed by GDLZ- formulas as shown Figure \ref{fig:nim}.

\begin{figure}
\centering
\begin{mdframed}

\begin{enumerate}
\item $initial \leftrightarrow turn(Player_1) \land \neg turn(Player_2) \land \langle \gamma_1, \cdot \cdot \cdot, \gamma_k \rangle$
\item $\bigwedge_{r\in N} wins(r) \leftrightarrow \neg turn(r) \land turn(-r)\land  \langle 0, \cdot \cdot \cdot, 0 \rangle$
\item $terminal \leftrightarrow \langle 0, \cdot \cdot \cdot, 0 \rangle $ 
\item $\bigwedge_{r\in N} \bigwedge_{m \in \{1...k\}}  \bigwedge_{s \in \{1...\gamma_m\}} 
legal(reduce^r(m,s)) \leftrightarrow  1 \leq s \leq heap_m $ 

\hspace{40pt}
$\land  turn(r)$

\item $\bigwedge_{r\in N} legal(noop^r)\leftrightarrow \neg turn(r)$

\item $\bigwedge_{i \in \{1...k\}}  \bigwedge_{h_i \in \{1...\gamma_i\}}  terminal  \text{ }\land$ $\langle h_1, \cdot \cdot \cdot, h_k \rangle \rightarrow \bigcirc \langle h_1, \cdot \cdot \cdot, h_k \rangle$
\item  $\bigwedge_{i \in \{1...k\}}  \bigwedge_{h_i \in \{1...\gamma_i\}} \neg terminal \land \langle h_1, \cdot \cdot \cdot, h_k \rangle$

\hspace{40pt}
$\land (\bigvee_{r\in N} \bigvee_{m \in \{1...k\}}  \bigvee_{s \in \{1...\gamma_m\}}   $ 
$does(reduce^r(m, s)) ) \rightarrow$	

\hspace{40pt}
$\bigcirc \langle h_1, \cdot \cdot \cdot,sub(h_m, s), \cdot \cdot \cdot, h_k \rangle 
$
\item $\bigwedge_{r\in N} turn(r) \rightarrow \bigcirc \neg turn(r) \land \bigcirc turn(-r)$
\end{enumerate}
\end{mdframed}
\caption{$\langle \gamma_1, \cdot \cdot \cdot, \gamma_k \rangle$-Nim Game represented by $\Sigma_{k\text{-nim}}$}
\label{fig:nim}
\end{figure}

Statement $1$ says that the $Player_1$ has the first turn and that the $k$ heaps starts with $ \gamma_1, \cdot \cdot \cdot, \gamma_k$ sticks, respectively. Statement $2$ and $3$ specify the winning states for each player and the terminal states of the game, respectively. The player who has not the turn when all the heaps become empty wins the game, and the game ends if all heaps are empty.  Statements $4$ and $5$ specify the preconditions of each action (legality). The player who has the turn can reduce $s$ sticks from the $m$-th heap if $1 \leq s \leq heap_m$. 
The other player can only do noop. Statements $6$ and $7$ define what is true at the next state: 
the $m$-th heap will be subtracted by $s$ if a player takes the action of reducing the $m$-th heap by $s$, otherwise it will keep its current value.  Finally, Statement $8$ specifies the turn-taking. Let $\Sigma_{k\text{-nim}}$ be the set of rules $1$-$8$.


%
Since the semantics for the language is based on the state transition model, 
we next specify the ST-model for this game, written $M_{k\text{-nim}}$, 
as follows:

\begin{itemize}
\item $W_{k\text{-nim}} = \{\langle t_1, t_2,\langle x_1, \cdot \cdot \cdot, x_k\rangle\rangle : t_1\in \{turn(Player_1), \neg turn(Player_1)\} \text{ \& }$ $t_2\in \{turn(Player_2), \neg turn(Player_2)\} \text{ \& } x_i \in \mathbb{N} \text{,  for }  1\leq i\leq k \}$ is the set of states, where $t_1, t_2$ specify the turn taking and $x_i$ represents the amount of sticks in the $i$-th heap, i.e. the integer value assigned to $heap_i$;
\item $\bar{w}_{k\text{-nim}} = \langle turn(Player_1), \neg turn(Player_2),\langle \gamma_1, \cdot \cdot \cdot, \gamma_k \rangle\rangle $; 
\item $T_{k\text{-nim}} = \{\langle turn(Player_1), \neg turn(Player_2), \langle 0, \cdot \cdot \cdot, 0\rangle \rangle,  \langle \neg turn(Player_1), $  \linebreak  $ turn(Player_2),$ $\langle 0, \cdot \cdot \cdot, 0\rangle \rangle \}$, i.e. all heaps are empty;

\item $L_{k\text{-nim}} = \{(\langle t_1, t_2,\langle x_1, \cdot \cdot \cdot, x_k\rangle\rangle,  reduce^r(m, s)) : t_r = turn(r) \text{ \& } 1 \leq s \leq x_m 
\}  
\cup \{ (\langle t_1, t_2,\langle x_1, \cdot \cdot \cdot, x_k\rangle\rangle,  noop^r) :  t_r = \neg turn(r)\}$, for all $\langle t_1, t_2,\langle x_1, \cdot \cdot \cdot, x_k\rangle\rangle \in W_{k\text{-nim}}$ and $r \in N_{k\text{-nim}}$; 

\item $U_{k\text{-nim}}: W_{k\text{-nim}} \times D_{k\text{-nim}} \rightarrow W_{k\text{-nim}}$ is defined as follows: for all $\langle t_1, t_2,\langle x_1, \cdot \cdot \cdot, x_k\rangle\rangle \in W_{k\text{-nim}}$ and all $( reduce^r(m,s), noop^{-r}) \in D_{k\text{-nim}}$, let $U_{k\text{-nim}}(\langle t_1, t_2,$ $\langle x_1, \cdot \cdot \cdot, x_k\rangle\rangle, ( reduce^r(m,s), noop^{-r} ) ) = \langle t_1', t_2',\langle x_1', \cdot \cdot \cdot, x_k'\rangle\rangle$, such that $\langle t_1', t_2',$ $\langle x_1', \cdot \cdot \cdot, x_k'\rangle\rangle$ are the same as $\langle t_1, t_2,\langle x_1, \cdot \cdot \cdot, x_k\rangle\rangle$, except by its components $t_1', t_2'$ and $x_i'$ which are updated as follows: 
$t_1' = turn(Player_1)$ iff $t_2 = turn(Player_2)$, otherwise $t_1' = \neg turn(Player_1)$; 
$t_2' = turn(Player_2)$ iff $t_1 = turn(Payer_1)$, otherwise, $t_2' = \neg turn(Player_2)$; 
 and for $1 \leq i \leq k$: 
\[x_i' = \begin{dcases}  
			x_i - s & \text{if }  reduce^r(i,s) \text{ and } 1 \leq s \leq x_i
            \\
            x_i & \text{otherwise} \\\end{dcases}\]
            
           For all  $\langle t_1, t_2,\langle x_1, \cdot \cdot \cdot, x_k\rangle\rangle \in W_{k\text{-nim}}$ and all $( a^r, a^{-r}) \neq ( reduce^r(m,s),$ \linebreak$ noop^{-r}) \in D_{k\text{-nim}}$, let $U_{k\text{-nim}}(\langle t_1, t_2,\langle x_1, \cdot \cdot \cdot, x_k\rangle\rangle, ( a^r, a^{-r} ) ) =\langle t_1, t_2,\langle x_1, \cdot \cdot \cdot, x_k\rangle\rangle$.

\item $g_{k\text{-nim}}(r) = \{\langle t_1, t_2, \langle  0,\cdot \cdot \cdot ,0 \rangle \rangle\}$, where $t_r  = \neg turn(r)$ and $t_{-r} = turn(r)$. 
\end{itemize}
Finally, for each state $w = \langle t_1, t_2,\langle x_1, \cdot \cdot \cdot, x_k\rangle\rangle\in W_{k\text{-nim}}$, let
\begin{itemize}
\item 
$\pi_{\Phi , k\text{-nim}}(w) = \{turn(r): t_r = turn(r)\}$; 
\item 
$\pi_{\mathbb{Z}, k\text{-nim}}(w) = \langle x_1, \cdot \cdot \cdot, x_k \rangle $. 
\end{itemize}

Let $M_{k\text{-nim}} = (W_{k\text{-nim}}, \bar{w}_{k\text{-nim}}, T_{k\text{-nim}}, L_{k\text{-nim}}, $ $U_{k\text{-nim}}, g_{k\text{-nim}}, \pi_{\Phi , k\text{-nim}}, $ $\pi_{\mathbb{Z}, k\text{-nim}})$ be the ST-model for the $k$-Nim Game.

\medskip

Consider, for instance, $k=2$ and $\langle \gamma_1, \gamma_2 \rangle = \langle 5, 3 \rangle$, i.e. there are only two heaps and their starting values are $5$ and $3$, respectively. 
Figure \ref{fig:flow-nim} illustrates a path in $M_{k\text{-nim}}$. The state $w_0$  represents the initial state. In $w_0$, it is the turn of $Player_1$ and he removes $5$ sticks from the first heap. In the state $w_1$, the first heap is empty and players can only remove sticks from the second heap. It is now $Player_2$'s turn and he reduces $2$ sticks from the second heap. In the state $w_2$,  $Player_1$ removes the last stick  from the second heap. Finally, in the state $w_3$, there is no stick remaining in any heap, thereby it is a terminal state. Since it is $Player_2$'s turn, $Player_1$ wins the game.

\begin{figure}[ht]
\centering
\tikzset{every picture/.style={line width=0.75pt}} 

\begin{tikzpicture}[x=0.4pt,y=0.4pt,yscale=-1,xscale=1]

\draw    (703.5,252) -- (754,252) -- (754,566) -- (706.5,566) ;
\draw [shift={(704.5,566)}, rotate = 360] [color={rgb, 255:red, 0; green, 0; blue, 0 }  ][line width=0.75]    (10.93,-3.29) .. controls (6.95,-1.4) and (3.31,-0.3) .. (0,0) .. controls (3.31,0.3) and (6.95,1.4) .. (10.93,3.29)   ;

\draw   (70,30) -- (360.5,30) -- (360.5,301) -- (70,301) -- cycle ;
\draw    (360.5,252) -- (377.5,252) -- (409,252) ;
\draw [shift={(411,252)}, rotate = 180] [color={rgb, 255:red, 0; green, 0; blue, 0 }  ][line width=0.75]    (10.93,-3.29) .. controls (6.95,-1.4) and (3.31,-0.3) .. (0,0) .. controls (3.31,0.3) and (6.95,1.4) .. (10.93,3.29)   ;

\draw    (70,202) -- (360.5,202) ;

\draw    (145.5,32) -- (145.5,301) ;

\draw    (70,125.5) -- (360.5,125.5) ;

\draw   (413,30) -- (703.5,30) -- (703.5,301) -- (413,301) -- cycle ;
\draw    (413,202) -- (703.5,202) ;

\draw    (488.5,32) -- (488.5,301) ;

\draw    (413,125.5) -- (703.5,125.5) ;

\draw   (70,343) -- (360.5,343) -- (360.5,614) -- (70,614) -- cycle ;
\draw    (70,515) -- (360.5,515) ;

\draw    (145.5,345) -- (145.5,614) ;

\draw    (70,438.5) -- (360.5,438.5) ;

\draw   (413,343) -- (703.5,343) -- (703.5,614) -- (413,614) -- cycle ;
\draw    (413,515) -- (703.5,515) ;

\draw    (488.5,345) -- (488.5,614) ;

\draw    (413,438.5) -- (703.5,438.5) ;

\draw    (413,566) -- (380.5,566) -- (362.5,566) ;
\draw [shift={(360.5,566)}, rotate = 360] [color={rgb, 255:red, 0; green, 0; blue, 0 }  ][line width=0.75]    (10.93,-3.29) .. controls (6.95,-1.4) and (3.31,-0.3) .. (0,0) .. controls (3.31,0.3) and (6.95,1.4) .. (10.93,3.29)   ;

\draw (555.5,80) node [scale=0.8]  {$ \begin{array}{l}
\neg initial\\
\neg terminal\\
\neg wins( Player_{1})\\
\neg wins( Player_{2})
\end{array}$};
\draw (450.5,67) node [scale=0.8,rotate=-270] [align=left] {{\small $\overline{w}_{k\text{\mbox{-}nim}}$ }\\{\small $T_{k\text{\mbox{-}nim}}$}\\{\small $g_{k\text{\mbox{-}nim}}$}};
\draw (554.5,164.5) node [scale=0.8]  {$ \begin{array}{l}
\neg turn( Player_{1})\\
turn( Player_{2})\\
\langle 0,3\rangle 
\end{array}$};
\draw (597,255) node [scale=0.8]  {$ \begin{array}{l}
legal(noop^{Player_{1}} )\\
legal(reduce^{Player_{2}}( m,s) )\\
\text{For } 1\leq m\ \leq 2,\ \\
\text{and } 1\ \leq s\ \leq heap_{m}
\end{array}$};
\draw (442.5,165) node [scale=0.8,rotate=-270] [align=left] {$ \begin{array}{l}
\pi _{\Phi ,k\text{\mbox{-}nim}} ,\\
\pi _{\mathbb{Z} ,k\text{\mbox{-}nim}}
\end{array}$};
\draw (433,238) node [scale=0.8,rotate=-270] [align=left] {$L_{k\text{\mbox{-}nim}}$};
\draw (543.5,321.5) node  [align=left] {$\displaystyle w_{1}$};
\draw (388,438.5) node [scale=0.8,rotate=-270]  {$ \begin{array}{l}
does(reduce^{Player_{1}}( 2,1) )\\
does(noop^{Player_{2}} )
\end{array}$};
\draw (212.5,393) node [scale=0.8]  {$ \begin{array}{l}
\neg initial\\
terminal\\
wins( Player_{1})\\
\neg wins( Player_{2})
\end{array}$};
\draw (107.5,380) node [scale=0.8,rotate=-270] [align=left] {{\small $\overline{w}_{k\text{\mbox{-}nim}}$ }\\{\small $T_{k\text{\mbox{-}nim}}$}\\{\small $g_{k\text{\mbox{-}nim}}$}};
\draw (211.5,477.5) node [scale=0.8]  {$ \begin{array}{l}
\neg turn( Player_{1})\\
turn( Player_{2})\\
\langle 0,0\rangle 
\end{array}$};
\draw (254,568) node [scale=0.8]  {$ \begin{array}{l}
legal(noop^{Player_{1}} )\\
legal(reduce^{Player_{2}}( m,s) )\\
\text{For } 1\leq m\ \leq 2,\ \\
\text{and } 1\ \leq s\ \leq heap_{m}
\end{array}$};
\draw (99.5,478) node [scale=0.8,rotate=-270] [align=left] {$ \begin{array}{l}
\pi _{\Phi ,k\text{\mbox{-}nim}} ,\\
\pi _{\mathbb{Z} ,k\text{\mbox{-}nim}}
\end{array}$};
\draw (90,551) node [scale=0.8,rotate=-270] [align=left] {$L_{k\text{\mbox{-}nim}}$};
\draw (200.5,634.5) node  [align=left] {$\displaystyle w_{3}$};
\draw (731,438.5) node [scale=0.8,rotate=-270]  {$ \begin{array}{l}
does(noop^{Player1} )\\
does(reduce^{Player_{2}}( 2,2) )
\end{array}$};
\draw (555.5,393) node [scale=0.8]  {$ \begin{array}{l}
\neg initial\\
\neg terminal\\
\neg wins( Player_{1})\\
\neg wins( Player_{2})
\end{array}$};
\draw (450.5,380) node [scale=0.8,rotate=-270] [align=left] {{\small $\overline{w}_{k\text{\mbox{-}nim}}$ }\\{\small $T_{k\text{\mbox{-}nim}}$}\\{\small $g_{k\text{\mbox{-}nim}}$}};
\draw (554.5,477.5) node [scale=0.8]  {$ \begin{array}{l}
turn( Player_{1})\\
\neg turn( Player_{2})\\
\langle 0,1\rangle 
\end{array}$};
\draw (597,568) node [scale=0.8]  {$ \begin{array}{l}
legal(reduce^{Player_{1}}( m,s) )\\
\text{For } 1\leq m\ \leq 2,\ \\
\text{and } 1\ \leq s\ \leq heap_{m}\\
legal(noop^{Player_{2}} )
\end{array}$};
\draw (442.5,478) node [scale=0.8,rotate=-270] [align=left] {$ \begin{array}{l}
\pi _{\Phi ,k\text{\mbox{-}nim}} ,\\
\pi _{\mathbb{Z} ,k\text{\mbox{-}nim}}
\end{array}$};
\draw (433,551) node [scale=0.8,rotate=-270] [align=left] {$L_{k\text{\mbox{-}nim}}$};
\draw (543.5,634.5) node  [align=left] {$\displaystyle w_{2}$};
\draw (200.5,321.5) node  [align=left] {$\displaystyle w_{0}$};
\draw (90,238) node [scale=0.8,rotate=-270] [align=left] {$L_{k\text{\mbox{-}nim}}$};
\draw (99.5,165) node [scale=0.8,rotate=-270] [align=left] {$ \begin{array}{l}
\pi _{\Phi ,k\text{\mbox{-}nim}} ,\\
\pi _{\mathbb{Z} ,k\text{\mbox{-}nim}}
\end{array}$};
\draw (254,255) node [scale=0.8]  {$ \begin{array}{l}
legal(reduce^{Player_{1}}( m,s) )\\
\text{For } 1\leq m\ \leq 2,\ \\
\text{and } 1\ \leq s\ \leq heap_{m}\\
legal(noop^{Player_{2}} )
\end{array}$};
\draw (211.5,164.5) node [scale=0.8]  {$ \begin{array}{l}
turn( Player_{1})\\
\neg turn( Player_{2})\\
\langle 5,3\rangle 
\end{array}$};
\draw (107.5,67) node [scale=0.8,rotate=-270] [align=left] {{\small $\overline{w}_{k\text{\mbox{-}nim}}$ }\\{\small $T_{k\text{\mbox{-}nim}}$}\\{\small $g_{k\text{\mbox{-}nim}}$}};
\draw (212.5,80) node [scale=0.8]  {$ \begin{array}{l}
initial\\
\neg terminal\\
\neg wins( Player_{1})\\
\neg wins( Player_{2})
\end{array}$};
\draw (388,125.5) node [scale=0.8,rotate=-270]  {$ \begin{array}{l}
does(reduce^{Player_{1}}( 1,5) )\\
does(noop^{Player_{2}} )
\end{array}$};

\end{tikzpicture}
\caption{A Path in $M_{k\text{-nim}}$, where $k=2$ and $\langle \gamma_1, \gamma_2 \rangle = \langle 5, 3 \rangle$} \label{fig:flow-nim}
\end{figure}
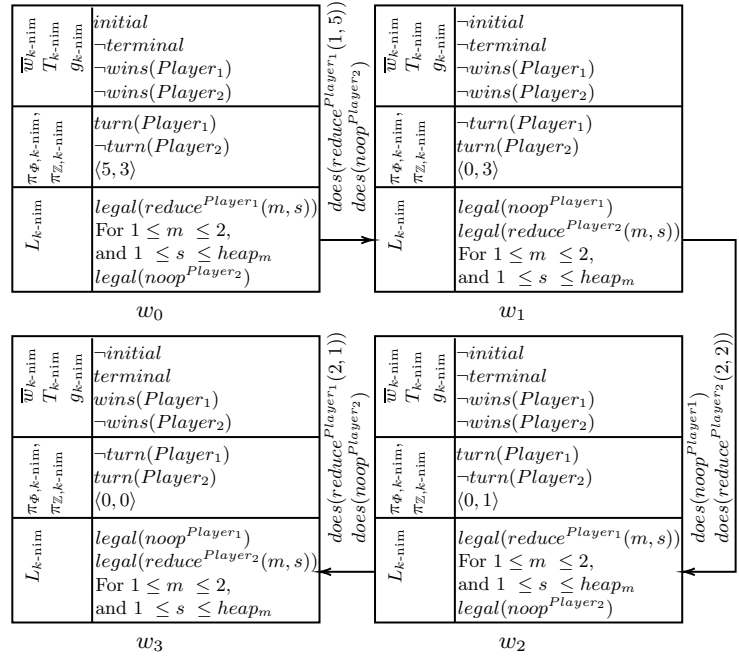

\end{example}



The next proposition shows that soundness does hold, i.e. the framework provides a sound description for the $k$-Nim Game. 
Notice that as $M_{k\text{-nim}}$ is not the unique model for $\Sigma_{k\text{-nim}}$, thereby, the completeness does not hold.


\begin{proposition}
\label{prop:modelknim}
$M_{k\text{-nim}}$  is an ST-model and it is a model of  $\Sigma_{k\text{-nim}}$.
\end{proposition}

\begin{proof}
It is routine to check that $M_{k\text{-nim}}$  is actually an ST-model. 
 Given any complete path $\delta$, any stage $t$ on $\delta$ in $M_{k\text{-nim}}$, we need to verify that each rule is true at $t$ of $\delta$ under $M_{k\text{-nim}}$.
 
Let us consider Rule $4$. 
Assume $M_{k\text{-nim}}, \delta, t \models \bigwedge_{m \in \{1...k\}}  \bigwedge_{s \in \{1...\gamma_m\}}  legal$ \linebreak $(reduce^r(m,s))$ iff  $reduce^r(m,s) $ $\in L_{k\text{-nim}}(\delta[t])$ iff $t_r = turn(r)$ and 
$1 \leq s \leq heap_{m}$ (by the definition of $L_{k\text{-nim}}$) iff $turn(r) \in \pi_{\Phi , k\text{-nim}}(\delta[t])$ 
(by the definition of $\pi_{\Phi , k\text{-nim}}$) 
and  $1 \leq s \leq x_{m}$, where $x_m$ is the value assigned to $heap_m$ at stage $\delta[t]$; 
iff $M_{k\text{-nim}}, \delta, t \models  1 \leq s \leq heap_m \land  turn(r) $. 

        
Let us verify Rule $7$. Assume $M_{k\text{-nim}}, \delta, t \models \bigwedge_{i \in \{1...k\}}  \bigwedge_{h_i \in \{1...\gamma_i\}} 
 \neg terminal \land \langle h_1, \cdot \cdot \cdot, h_k \rangle \land \bigvee_{r\in N} \bigvee_{m \in \{1...k\}}  \bigvee_{s \in \{1...\gamma_m\}}  does(reduce^r(m, s))$ . 
Since $\neg terminal$, by the path definition, we have $t <|\delta|$.
For some $r \in N$, $1 \leq m \leq k$ and $1 \leq s \leq \gamma_m$, it is true that $does(reduce^r(m, s))$, then $\theta_r(\delta, t) = reduce^r(m, s) \in L_{k\text{-nim}}(\delta[t])$. Since $M_{k\text{-nim}}, \delta, t \models \langle h_1, \cdot \cdot \cdot, h_k \rangle$, for some $h_i,\in \{1, \cdot \cdot \cdot, \gamma_i\}$ and any $i \in \{1, \cdot \cdot \cdot, k\}$, by the definition of $U_{k\text{-nim}}$, we have $M_{k\text{-nim}}, \delta, t + 1\models  \langle x_1', \cdot \cdot\cdot, x_k'\rangle $, where $x_{m}' = heap_{m} - s$ if $reduce^r(m, s)$ and for any $j\neq m$ and $1 \leq j \leq k$, $x_j' = heap_j$. 
By function $v$, we know that $v(sub(h_{m},s)) = heap_{m} - s$.
Thereby, $M_{k\text{-nim}}, \delta, t+1 \models \langle h_1, \cdot \cdot \cdot,sub(h_m,$ $ s), \cdot \cdot \cdot, h_k \rangle $ and so
 so $M_{k\text{-nim}}, \delta, t \models \bigcirc \langle h_1, \cdot \cdot \cdot,sub(h_m, s), \cdot \cdot \cdot, h_k \rangle $.

The remaining rules are proved in a similar way.
\end{proof}

In the next section, we show that the model checking for GDLZ is PTIME, which is the same complexity then the model checking for GDL. 
In other words, the addition of numerical features in GDL does not increase the complexity at verifying the validity of a formula at a stage of a path in a model. 

\subsection{Model Checking}
\label{sec:Model-Checking}

The \textit{model checking problem} for GDLZ 
is the following: Given a GDLZ-formula~$\varphi$, an ST-model $M$, a path $\delta$ of $M$ and a stage $j$ on $\delta$,  
determine whether $M, \delta, j \models \varphi$ or not.  

Let $Sub(\varphi)$ be the set of all subformulas\footnote{We say that $\psi$ is a subformula of $\varphi \in \mathcal{L}_{GDLZ}$ if either (i) $\psi = \varphi$; (ii) $\varphi$ is of the form $\neg \varphi'$ or $\bigcirc \varphi '$ and $\psi$ is a subformula of $\varphi '$; or (iii) $\varphi$ is of the form $\varphi' \land \varphi''$ and $\psi$ is a subformula of either $\varphi'$ or $\varphi''$. } 
of $\varphi$.
Algorithm \ref{alg:MC} works in the following way: first it gets all subformulas of $\varphi$ and orders them in $S$ by its ascending length. Thus, $S(|\varphi|) = \varphi$, i.e. the position $|\varphi|$ in the vector $S$ corresponds to the formula $\varphi$ itself,  and if $\phi_i$ is a subformula of $\phi_j$, then $i < j$. An induction on $S$ label each subformula $\phi_i$ depending on whether or not $\phi_i$ is true in $M$ at $\delta[j]$.  If $\phi_i$ does not have any subformula, its truth value is obtained directly from the semantics. 
Since $S$ is ordered by the formulas length, if $\phi_i$ 
is either in the form $\phi' \land \phi''$ or $\neg \phi'$ the algorithm labels $\phi_i$ according to 
the label assigned to $\phi'$ and/or $\phi''$. If $\phi_i$ is in the form $\bigcirc \phi'$, its label will be recursively defined according to $\phi'$ truth value in $\delta[j+1]$.
As Algorithm \ref{alg:MC} visits each node at most once, and the number of nodes in the tree is not greater than the size of $\varphi$, 
it can be clearly
implemented in a polynomial-time deterministic Turing machine with PTIME.

\begin{algorithm}
\caption{$isTrue(M, \delta, j,  \varphi)$}\label{alg:MC}
\hspace*{\algorithmicindent} \textbf{Input:} an ST-model $M$, a path $\delta$ of $M$, a stage $j$ and 
a  formula $\varphi \in  \mathcal{L}_{GDLZ}$.
 \\
 \hspace*{\algorithmicindent} \textbf{Output:} \textbf{true} if $M, \delta, j \models \varphi$, and \textbf{false} otherwise
\begin{algorithmic}[1]
\State $S \gets $ $Sub(\varphi)$ ordered by ascending length
\State Let $reg[1 \cdot \cdot \cdot size(S)]$ be a boolean array

\For {$i \gets 1$ to $size(S)$} 
\State $\phi \gets S[i]$
\If {($\phi = \phi'\land \phi''$)}
	\State $reg[i] \gets reg[getIndex(S, \phi')] \land reg[getIndex(S, \phi'')]$
\ElsIf {($\phi = \bigcirc\phi'$)}
	\State $reg[i] \gets isTrue(M, \delta, j+1, \phi')$
\ElsIf {($\phi = \neg\phi'$)}
	\State $reg[i] \gets \neg  reg[getIndex(S, \phi')]$
\Else $\text{ } reg[i] \gets M, \delta, j \models \phi$ 
\EndIf
\EndFor
\Return $reg[size(S)] $
\end{algorithmic}
\end{algorithm}


In Section \ref{sec:GDLTranslation} we show that $\mathcal{L}_{GDL} \subseteq \mathcal{L}_{GDLZ}$, i.e. any formula in GDL is also a formula in GDLZ. Thereby, Algorithm \ref{alg:MC} can also be used in the model checking problem for GDL.

\section{Translation Between GDLZ and GDL}
\label{sec:GDLxZ}

In this section, we investigate translation maps among GDLZ and GDL models and descriptions.  
We first consider the general case where the GDLZ ST-model can have infinite components. Next, we restrict to the case where a GDLZ ST-model is finite.
Finally, we compare both languages in order to show the succinctness of GDLZ descriptions over GDL descriptions.

Given a GDLZ ST-model $M$, a complete path $\delta$ in $M$ 
and a formula $\varphi \in \mathcal{L}_{GDLZ}$, 
in the Sections \ref{sec:ModelOverPathTranslation} and \ref{sec:ModelTranslation} our goal is 
to construct a GDL ST-model $M'$, a path $\delta'$ in $M'$ and a formula $\varphi' \in \mathcal{L}_{GDL}$ such that,  for any stage $j$ on $\delta$, if $M, \delta, j \models \varphi$ then $M', \delta', j \models \varphi'$.


\subsection{From GDLZ Paths and Models to GDL Models}
\label{sec:ModelOverPathTranslation}


In a GDL ST-model,  the sets of states, actions and atomic propositions are finite. Since it does not hold for GDLZ ST-models, it is not possible to define a complete translation from every GDLZ model to a GDL model. However, since any GDLZ path is a finite sequence of states and joint actions, we can define  a partial translation from GDLZ ST-models to GDL ST-models based on the reached states and joint actions performed in a complete path. In other words, we can translate a run in a GDLZ model into a GDL model. Let us formally describe the translation.



Through the rest of this section, we fix the GDLZ ST-model
$M = (W, \bar{w}, T, $ $L, U, g, \pi_\Phi , \pi_{\mathbb{Z}})$ with a game signature $\mathcal{S} = (N, \mathcal{A}, X, \Phi )$ and the complete path $\delta = \bar{w}\stackrel{d_1}{\to} w_1 \stackrel{d_2}{\to}\cdot\cdot\cdot \stackrel{d_e}{\to} w_e$ in $M$.

Given the path $\delta$ in $M$, we next define a shortcut to refer to the smallest and biggest integers occurring in $\delta$ and the set of all actions performed in $\delta$. 


\begin{definition}
\label{def:minmax}
Given $M$ and $\delta$, 
we denote $\delta_{min}$ and $ \delta_{max}$ as the smallest  and biggest integer, respectively, occurring in any parameter list $\mathrm{z}$ from any action $a \in \{d_1, d_2, \cdot \cdot \cdot,  d_e\}$ and in any $\pi_\mathbb{Z}(w)$, for $w \in \{\bar{w}, w_1, \cdot \cdot \cdot, d_e\}$.
\end{definition}

\begin{definition}
\label{def:actioninpath}
Given $M$ and $\delta$, 
let $\mathcal{A}^\delta = \{d_j(r):  r \in N \text{ \& } 1 \leq j \leq e\}$ denote the set of all actions performed in $\delta$.
\end{definition}


Since we are aware of the path numerical range, we are able to construct a partial model translation. The translation is restricted to the states and actions involved in a given path. %

\begin{definition}
\label{def:modeltr}
Given a GDLZ ST-model $M$ and $\delta$,
we construct an associated GDL ST-model $M' = (W', \bar{w}, T', $ $L', U', g', \pi')$ with a game signature $\mathcal{S}' = (N, \mathcal{A}', \Phi')$. The components $\bar{w}$ and $N = \{r_1, \cdot \cdot \cdot , r_\mathsf{k}\}$ are the same for $M$ and $M'$. 

The propositional set $\Phi'$ is constructed over both $\Phi$ and $X$ 
as follows: $\Phi' = \{p, smaller(z_1, z_2), bigger(z_1, z_2), equal(z_1, z_2), succ(z_1, z_2), prec(z_1, z_2), x(q):p $ $\in \Phi, x \in X, \delta_{min} \leq q, z_1, z_2 \leq \delta_{max}\}$. The notation $x(q)$ represents the proposition ``variable $x$ has the value $q$".

For integrating the GDLZ comparison operators $<, >$ and $=$ in GDL, we need to define the order between the numerical terms in the translated model. 
Let $\pi_{z} \subset \Phi'$ denote a set of propositions  describing the numerical order, such as:
%
$\pi_z = \{succ(z, z+1), prec(z+1, z),  equal(z_1, z_1) : \delta_{min} \leq z < \delta_{max} 
\text{ \& } \delta_{min} \leq z_1 \leq \delta_{max}\} $ $\cup$ $ \{smaller(z_1, z_2): \delta_{min}\leq z_1, z_2 \leq \delta_{max} \text{ \& } z_1<z_2\} $ $\cup$ $ \{bigger(z_1, z_2): \delta_{min}\leq z_1, z_2 \leq \delta_{max} \text{ \& } z_1>z_2\}$.     


For any $a^r( z_1, \cdot \cdot \cdot, z_{l}) \in \mathcal{A}^\delta$, $a^r_{z_1, \cdot \cdot \cdot, z_{l}} \in \mathcal{A}'$.
We define an action translation 
$Tr^a: \mathcal{A}^\delta \rightarrow \mathcal{A}'$ associating every action in $\mathcal{A}^\delta$ with an action in $\mathcal{A}'$: 
\[Tr^a(a^r( z_1, \cdot \cdot \cdot, z_{l} )) = a^r_{z_1, \cdot \cdot \cdot, z_{l}}\] 
where $a^r( z_1, \cdot \cdot \cdot, z_{l}) \in \mathcal{A}^\delta$. 

The $M'$ components  $W',  T', L', U', g' $ and $\pi'$  are defined as follows:

\begin{itemize}
\item $W' = \{\bar{w}, w_1, \cdot \cdot \cdot, w_e\}$ 
\item $T' = \{w_e\}$;
\item $L' = \{(w_{j-1}, Tr^a(d_j(r)) : r \in N \text{ \& } 1 \leq j \leq e\}$;
\item $U'(w_{j-1}, ( Tr^a(d_j(r_1)), \cdot \cdot \cdot,  Tr^a(d_j(r_\mathsf{k})) )) = w_j$, for $1 \leq j \leq e\}$ ; 
\item $g'(r) = \{\{w_e \}\}$ if $ w_e \in g(r)$, otherwise $g'(r) = \emptyset$, for $r\in N$;
\item $\pi'(w) = \{ \pi_\Phi(w)\} \cup \{\pi_z\} \cup \{x(q): q\in\pi_\mathbb{Z}(w),
x \in X
\}$, for $w \in W'$.
\end{itemize}

We say that $M' = (W', \bar{w}', T', L', U', g', \pi')$ with the  
signature $\mathcal{S}' = (N, \mathcal{A}', $ $\Phi')$ 
is the ST-model translation of $M$ restricted over $\delta$ and write $Tr^m(M,  \delta)$.


\end{definition}

The path translation assigns each action appearing on it to the appropriated GDL action through $Tr^a$, i.e. the action translation.

\begin{definition}
\label{def:pathtr}
Given the agent set $N = \{r_1, \cdot \cdot \cdot, r_\mathsf{k}\}$, define a path translation $Tr^\lambda: \mathcal{\delta} \rightarrow \mathcal{\delta}'$ associating a path $\delta = \bar{w}\stackrel{d_1}{\to} w_1 \stackrel{d_2}{\to}\cdot\cdot\cdot \stackrel{d_e}{\to} w_e$ in $M$ with a path $\delta'$ in $Tr^m(M, \delta)$:
$
Tr^\lambda(\delta) = \bar{w}\stackrel{d_1'}{\to} w_1 \stackrel{d_2'}{\to}\cdot\cdot\cdot \stackrel{d_e'}{\to} w_e $,
where $d_i' = ( Tr^a(d_i(r_1)), \cdot \cdot \cdot, Tr^a(d_i(r_\mathsf{k})) )$,  for 
$1\leq i \leq e$.

\end{definition}

As shown next propositions, given a path in a GDLZ model, the translation of the GDLZ model is a GDL model. Moreover, the translation of a path in a GDLZ model is a path in the translation of the GDLZ model.  

\begin{proposition}
\label{prop:trmodel}
If $M$ is a GDLZ model and $\delta$ a complete path in $M$, then $Tr^m(M, \delta)$ is a GDL ST-model.
\end{proposition}

\begin{proof}
Given $M$ and $\delta$, 
let $Tr^m(M, \delta) = (W', \bar{w}, T', L',$ $ U', g', \pi')$, with $\mathcal{S'}= (N, \mathcal{A}', \Phi')$.
Since $\delta$ is a finite sequence of states and joint actions,  
we have that $W', \mathcal{A'}$ and $\Phi'$ are ensured to be finite sets. 
 
Since $Tr^a: \mathcal{A}^\delta\rightarrow \mathcal{A}'$ is an injective function, 
each $a \in \mathcal{A}^\delta$ will be assigned to a unique  $a' \in A'$.  By the path definition, we know that $d_j(r) \in L(w_{j-1})$, for every $r \in N,  1\leq j \leq e$. Then, it is easy to see that $L' \subseteq W' \times \mathcal{A}'$. By $Tr^m$ definition, we know that $\bar{w} \in W'$, $T' \subseteq W'$ and $g'(r) \subseteq \{\{w_e\}, \emptyset\}$, thereby $g'(r) \subseteq  2^{W'}$, for $r \in N$.
Furthermore, for every stage $1 \leq j \leq e$, we have $U'(w_{j-1},(Tr^a(d_j(r_1)), \cdot \cdot \cdot, Tr^a(d_j(r_\mathsf{k})) )) = U(w_{j-1}, (d_j(r_1), \cdot \cdot \cdot, d_j(r_\mathsf{k}) ))$, thus $U'(w_{j-1},(Tr^a(d_j(r_1)), \cdot \cdot \cdot, Tr^a(d_j(r_\mathsf{k})) ))  \in W'$.  Finally, since $Tr^m$ defines 
$\Phi' = \{p, smaller(z_1, z_2), bigger(z_1, z_2), $ $equal(z_1, z_2), $ $succ(z_1, z_2), prec(z_1, z_2),$ $ x(q): p \in \Phi, x \in X, \delta_{min} \leq q, z_1, z_2 \leq \delta_{max}\}$, 
then for every $w \in W'$, we have that $\pi'(w) \in \{ \pi_\Phi(w) \cup \pi_z \cup \{x(q): q\in\pi_\mathbb{Z}(w),
x \in X
\}\}$ and thus  $\pi'(w) \subseteq 2^{\Phi'}$.
Therefore, $Tr^m(M, \delta)$ is a GDL ST-model. 
\end{proof}

\begin{proposition}
\label{prop:trpath}
If $\delta$ is a path in a GDLZ model M then $Tr^\lambda(\delta)$ is a path in $Tr^m(M, \delta)$.
\end{proposition}

\begin{proof}
Given $M$, $\delta$ and $Tr^m(M, \delta) = (W', \bar{w}, T', L',$ $ U', g', \pi')$ with $\mathcal{S}= (N, \mathcal{A}', $ $\Phi')$. 
Then $Tr^\lambda(\delta) = \bar{w}\stackrel{d_1'}{\to} w_1 \stackrel{d_2'}{\to}\cdot\cdot\cdot \stackrel{d_e'}{\to} w_e$. 
By the GDLZ path definition, for $e\geq 0$ and for any $j \in \{1, \cdot\cdot\cdot, e\}$, we have $\{w_0, \cdot\cdot\cdot, w_{e-1}\} \cap T = \emptyset$, where $w_0 = \bar{w}$. 
Since $T' = \{w_e\}$, we have  $\{w_0, \cdot\cdot\cdot, w_{e-1}\} \cap T' = \emptyset$.

For any $r \in N$, we have that $d_j(r) \in L(w_{j-1})$. Since the action translation $Tr^a$ assigns each action in $\mathcal{A}^\delta = \{d_j(r):  r \in N \text{ \& } 1 \leq j \leq e\}$ to an unique action in $\mathcal{A}'$, then the translation from the action of agent $r$ in the joint action $d_j$ will be in
the set of the translated legal actions in state $w_{j-1}$, i.e. $Tr^a(d_j(r)) \in L'(w_{j-1})$, where 
$L'(w_{j-1}) = \{Tr^a(a) \in \mathcal{A}^\delta \mid (w_{j-1},Tr^a(a)) \in L' \}$. 
Thus, $Tr^a(d_j(r)) \in L'(w_{j-1})$. 
Finally, since $\delta$ is path, $w_j = U(w_{j-1}, d_j) = U(w_{j-1}, (d_j(r_1), \cdot \cdot \cdot, d_j(r_\mathsf{k}) )$. Then, $w_j = U'(w_{j-1}, ( Tr^a(d_j(r_1)), \cdot \cdot \cdot, Tr^a(d_j(r_\mathsf{k}) ) )$,  that is, $w_j = U'(w_{j-1}, (a^{r_1}\text{}', \cdot \cdot \cdot, a^{r_\mathsf{k}}\text{}' )) = U'(w_{j-1}, d_j')$.

Thus, we have that $Tr^\lambda(\delta) = \bar{w}\stackrel{d_1'}{\to} w_1 \stackrel{d_2'}{\to}\cdot\cdot\cdot \stackrel{d_e'}{\to} w_e$ is a path in the GDL ST-model $Tr^m(M, \delta)$.
Furthermore, if $\delta$ a complete path in $M$, then $w_e \in T$ and $Tr^\lambda(\delta)$ is also a complete path in $Tr^m(M, \delta)$. 
\end{proof}

Next, we show how to translate GDLZ formulas to GDL. Likewise to the model translation, the translation is restricted to a path. 

\subsubsection{From GDLZ Paths and Formulas to GDL Formulas.}
\label{sec:DescriptionsTranslation}

Let us briefly recall GDL grammar. Given a GDL game signature $\mathcal{S}' = (N, \mathcal{A}', \Phi')$, a formula $\varphi' \in \mathcal{L}_{GDL}$ is defined by the following BNF:
\begin{align*}
\varphi' ::= p \mid initial \mid terminal \mid legal(a^r)\mid wins(r) \mid does(a ^r)  \mid \neg \varphi \mid \varphi \land \varphi \mid 
\bigcirc \varphi
\end{align*}
where $p \in \Phi'$, $r \in N$ and $a^r \in \mathcal{A}'$.



Given a path $\delta$ in a GDLZ ST-model $M$, we next define a translation for formulas in $\mathcal{L}_{GDLZ}$ to $\mathcal{L}_{GDL}$.
Each numerical term $z \in \mathcal{L}_z$ occurring in a formula $\varphi \in \mathcal{L}_{GDLZ}$ is translated by its semantic interpretation through function $v$ (see Definition \ref{def:funcionv}).  

 

\begin{definition}
\label{def:formulatr}
Given a GDLZ ST-model $M$ with $\mathcal{S} = (N, \mathcal{A}, X, \Phi )$, a path $\delta$ in $M$, a stage $j$ in $\delta$ and  function $v$ (see Definition \ref{def:funcionv}). 
A translation $Tr^\varphi$ from a formula $\varphi \in \mathcal{L}_{GDLZ}$ in a state $\delta[j]$ to a formula $\varphi' \in \mathcal{L}_{GDL}$ is defined as follows:
 
\begin{itemize}
\item $Tr^\varphi(\varphi, \delta[j]) = \varphi$ for all $\varphi \in \Phi \cup \{initial, terminal, wins(r)\}$;
\item $Tr^\varphi(\neg \varphi, \delta[j]) = \neg Tr^\varphi(\varphi, \delta[j])$;
\item $Tr^\varphi(\varphi_1 \land \varphi_2, \delta[j]) = Tr^\varphi(\varphi_1, \delta[j]) \land Tr^\varphi(\varphi_2 , \delta[j])$;
\item $Tr^\varphi(\bigcirc \varphi, \delta[j]) = \bigcirc Tr^\varphi(\varphi, \delta[j+1])$;
\item $Tr^\varphi(legal(a^r(\bar{z})), \delta[j]) = legal(Tr^a (a^r(v(z) : z \in \bar{z} )))$ iff $legal(a^r(v(z, \delta[j]): z \in \bar{z}) = \theta_r(\delta, j)$; otherwise $Tr^\varphi(legal(a^r(\bar{z})), \delta[j])= \neg legal(Tr^a (a^r(v(z) : z \in \bar{z} )))$;
\item $Tr^\varphi(does(a^r(\bar{z})), \delta[j]) = does(Tr^a(a^r(v(z) : z \in \bar{z} )))$;
\item 
$Tr^\varphi(\langle\bar{z}\rangle, \delta[j]) = \bigwedge_{i = 1}^{|\bar{z}|} x_i(v(q_i, \delta[j]))$;
\item $Tr^\varphi(z_1 < z_2, \delta[j]) = smaller(v(z_1, \delta[j]), v(z_2, \delta[j]))$;
\item $Tr^\varphi(z_1 > z_2, \delta[j]) =  bigger(v(z_1, \delta[j]), v(z_2, \delta[j]))$;
\item $Tr^\varphi(z_1 = z_2, \delta[j]) = equal(v(z_1, \delta[j]), v(z_2, \delta[Tr^\varphi(\varphi, \delta[t])]))$.
\end{itemize}

Where $r \in N$, $x_i \in X, q_i$ is the $i$-th value in $ \bar{z}$ and $ 0\leq i \leq |\bar{z}|$.


\end{definition}


Given a path in a GDLZ model, we show that the translation of a GDLZ formula is a GDL formula. Furthermore, if the GDLZ formula is valid at a stage in the path, its translation will be valid at the same stage in the translated path in the translated model.

\begin{proposition}
\label{prop:trformula}
Given a GDLZ ST-model $M$, a path $\delta$ in $M$, a stage $j$ in $\delta$ and function $v$, if $\varphi$ is a formula in $\mathcal{L}_{GDLZ}$  then $Tr^\varphi(\varphi, \delta[j])$ is a formula in $\mathcal{L}_{GDL}$.
\end{proposition}

\begin{proof}
Given a GDLZ model $M  = (W, \bar{w}, T, L, U,g,\pi_{\Phi}, \pi_{\mathbb{Z}})$, with a game signature $\mathcal{S} = (N, \mathcal{A}, \Phi, X)$, a path $\delta$ in $M$, a stage $j$ in $\delta$ and function $v$. 
Let $M' = Tr^m(M, \delta)$, with $\mathcal{S}' = (N, \mathcal{A}', \Phi')$.  
Assume that $\varphi \in \mathcal{L}_{GDLZ}$, 
we show that $Tr^\varphi(\varphi, \delta[j]) \in \mathcal{L}_{GDL}$ for each form of $\varphi$:

\begin{itemize}
\item If $\varphi$ is of the form $p, initial, terminal , wins(r),  \neg \varphi, \varphi \land \varphi$ or $ 
\bigcirc \varphi$, where $p \in \Phi$ and $r \in N$, then $Tr^\varphi(\varphi, \delta[j])$ assigns $\varphi$ to the exactly corresponding $\varphi' \in \mathcal{L}_{GDL}$. Thus, $Tr^\varphi(\varphi, \delta[j]) \in \mathcal{L}_{GDL}$.
\item If $\varphi$ is of the form $legal(a^r(\bar{z}))$ or $does(a^r(\bar{z}))$, where $r \in N$, then $Tr^\varphi(\varphi, \delta[j])$ $ =  legal(Tr^a(a^r(v(z) : z \in \bar{z} )))$ or $Tr^\varphi(\varphi, \delta[j]) =  does(Tr^a(a^r(v(z) : z \in \bar{z} )))$, respectively.  Since $Tr^a$is an injective function from $\mathcal{A}^\delta$ to $\mathcal{A}'$, we have that $Tr^a(a^r) = a^r\text{}'\in \mathcal{A}'$. 
		Therefore, $legal(a^r\text{}'), does(a^r\text{}') \in \mathcal{L}_{GDL}$ and $Tr^\varphi(\varphi, \delta[j]) $ $\in \mathcal{L}_{GDL}$.
\item If $\varphi$ is of the form $z_1 < z_2, z_1  >z_2, $ or $z_1 = z_2$, since $\{smaller(z_1, z_2), $\linebreak $ bigger(z_1, z_2)$, $equal(z_1, z_2), \delta_{min} \leq q, z_1, z_2 \leq \delta_{max}\} \in \Phi'$ and 
$\varphi_1 \land \varphi_2 \in \mathcal{L}_{GDL}$, we have that $\varphi \in \mathcal{L}_{GDL}$.
\item Finally, if $\varphi$ is of the form $\langle\bar{z}\rangle$, then $Tr^\varphi(\varphi, \delta[j]) = \bigwedge_{i = 1}^{|\bar{z}|} x_i(v(q_i, \delta[j]))$, 
where $x_i \in X$, $q_i$ is the $i$-th value of $\bar{z}$ and $0\leq i \leq |\bar{z}|$. We have that 
$\{x(q): x \in X, \delta_{min} \leq q \leq \delta_{max}\} \subseteq \Phi' $. Since that for each $p \in \Phi'$, $p \in \mathcal{L}_{GDL}$, we have that each $x_i(q_i) \in \mathcal{L}_{GDL}$.  Moreover, for any $\varphi_1$, $\varphi_2 \in \mathcal{L}_{GDL}$, we also have $\varphi_1 \land \varphi_2 \in \mathcal{L}_{GDL}$, then$((x_1(q_1) \land (x_2(q_2))\cdot \cdot \cdot ) \land x_{|\bar{z}|}(q_{|\bar{z}|}) = x_1(q_1) \land \cdot \cdot \cdot  \land x_{|\bar{z}|}(q_{|\bar{z}|}) = Tr^\varphi(\varphi, \delta[j])\in \mathcal{L}_{GDL}$.  
\end{itemize}
\end{proof}

\begin{theorem}
\label{prop:valid}
If $M, \delta, j \models\varphi$  then $Tr^m(M, \delta), Tr^\lambda(\delta), j \models Tr^\varphi(\varphi, \delta[j])$ . 
\end{theorem}

\begin{proof}
Given a GDLZ model $M = (W, \bar{w}, T, L, U, g, \pi_\Phi , \pi_{\mathbb{Z}})$, with the game signature $\mathcal{S} = (N, \mathcal{A}, X, \Phi )$, a complete path $\delta$, a stage $j$ on $\delta$, a formula $\varphi \in \mathcal{L}_{GDLZ}$ and the function $v$.
Let $M' = (W', \bar{w}, T', L',$ $ U', g', \pi')$, with $\mathcal{S'}= (N, \mathcal{A}', \Phi')$, be the GDL translation of $M$, i.e.  $M' = Tr^m(M, \delta)$,  $\delta' = Tr^\lambda(\delta)$ 
and  $\delta_{min}$, $ \delta_{max} \in \mathbb{Z}$ denote the integer bounds in $\delta$.

For any integers $\delta_{min} \leq z_1, z_2 < \delta_{max}$, $\pi_z \subseteq \pi'(\delta[j])$ enumerates its predecessor and successor and define all the cases were $bigger(z_1, z_2)$, $smaller(z_1, z_2)$ and $equal(z_1, z_2)$ are true. 
Let $\varphi' = Tr^\varphi(\varphi, \delta[j])$. 
We assume that $M, \delta, j \models \varphi$ and show that then we have $M', \delta', j \models\varphi'$ for every $\varphi$.

\begin{itemize}
\item  If $\varphi $ is on the form $ p \in \Phi$, we have $Tr^\varphi(p, \delta[j]) = p$. By $\mathcal{L}_{GDLZ}$ semantics, we know that $p \in \pi_\Phi (\delta [j])$. In the ST-model translation, we have the valuation function constructed such that $\pi'(\delta[j]) = \{ \pi_\Phi(\delta[j])\} \cup \{\pi_z\}\cup \{x(q): q\in\pi_\mathbb{Z}(\delta[j]), x \in X \}$
. Then, $p \in \pi'(\delta[j]')$ and $M', \delta', j \models p$;
\item If $\varphi$ is either on the form $\neg \psi$, $\varphi_1 \land \varphi_2$, $initial$, $terminal$, $wins(r)$, $legal(a^r(\bar{z}))$, $does(a^r(\bar{z}))$, or $\bigcirc \psi$, since $Tr^a$ and $Tr^\varphi$ assigns each GDLZ action and formula to an unique GDL state, action and formula, respectively, due to both languages semantics it is easy to see that $M', \delta', j \models Tr^\varphi(\varphi, \delta[j])$, whenever $M, \delta, j \models \varphi$;
\item  If $\varphi$ is on the form $z_1 > z_2$, we have $Tr^\varphi(z_1 > z_2, \delta[j]) =  bigger(v(z_1, \delta[j]), $ $v(z_2, \delta[j]))$. By $\mathcal{L}_{GDLZ}$ semantics, we know that $  v(z_1, \delta[j]) > v(z_2, \delta[j]) $, i.e. $v(z_1, \delta[j])$ is bigger then $v(z_2,\delta[j])$, then $bigger(v(z_1, \delta[j]), v(z_2, \delta[j])) \in \pi_z$. 
$\pi_z\subseteq \pi'(\delta[j])$ defines $ bigger(v(z_1, \delta[j]), $ $ v(z_2, \delta[j]))$ such that it is true, iff $v(z_1, \delta[j]) > v(z_2, \delta[j])$. 
Thus, $M', \delta', j \models bigger(v(z_1, \delta[j]), v(z_2, \delta[j])) $;
\item If $\varphi$ is either on the form $z_1 < z_2$ or $z_1 = z_2$, the proof proceeds as in the previous case;
\item If $\varphi$ is on the form $\langle\bar{z}\rangle $, $Tr^\varphi(\bar{z}, \delta[j]) =\bigwedge_{i = 1}^{|\bar{z}|} x_i(v(q_i, w))$, where $x_i \in X$ and $q_i$ is the $i$-th value of $\bar{z}$. By $\mathcal{L}_{GDLZ}$ semantics, we know that  $\bar{z} = \pi_\mathbb{Z}(\delta[j]) $. Since, by the ST-model translation each $x_i(q_i) \in \Phi'$ and $\pi'(\delta[j]) = \{ \pi_\Phi(\delta[j])\} \cup \{\pi_z\}\cup \{x(q): q\in\pi_\mathbb{Z}(\delta[j]), x \in X \}$, we have that $M', \delta', j \models x_1(q_1)$, $M', \delta', j \models x_2(q_2)$ and so on, thus $M', \delta', j \models \bigwedge_{i = 1}^{|\bar{z}|} x_i(v(q_i, \delta[t]))$. 
\end{itemize}
\end{proof}

Because it is a partial translation based on a path, the legal actions are restricted to the ones performed in the path. To overcome this issue, in the next section we show how to define complete translations over GDLZ models and formulas. The following complete translation is limited to the finite GDLZ models.

\subsection{From Finite GDLZ Model to GDL Model}
\label{sec:ModelTranslation}


Let us consider the case where the GDLZ ST-model has finite components. In this case, we are able to define a complete model translation, instead of partial based on a path. In other words, all possible runs over the finite GDLZ ST-model can be translated. 
Next, we characterize a finite GDLZ ST-model. 

\begin{definition}
\label{def:finitemodeltr}
Given two arbitrary bounds $z_{min} \leq  z_{max} \in \mathbb{Z}$,
a finite GDLZ ST-model $M_f = (W_f, \bar{w}_f, T_f, L_f, U_f, g_f, \pi_{\Phi f} , \pi_{\mathbb{Z} f})$, with the game signature $\mathcal{S}_f = (N_f, \mathcal{A}_f, X_f, $ $\Phi_f )$ is a subset of GDLZ ST-models that have the following aspects: (i) $z_{min} \leq z_i \leq z_{max} $, for any $ a^r( z_1, \cdot\cdot\cdot, z_{l}) \in \mathcal{A}_f   $, $1\leq i \leq o$ and $ r \in N_f$; (ii) $W_f$ and $\mathcal{A}_f$ are finite sets; and (iii) $z_{min} \leq q_i \leq z_{max} $, for any $\langle q_1 \cdot \cdot \cdot q_n \rangle = \pi_\mathbb{Z}(w)$, $1 \leq i\leq n$ and $w \in W_f$.
\end{definition}

Through the rest of this section, we fix the bounds $z_{min}$ and $z_{max}$ as well as the finite GDLZ ST-model $M_f = (W_f, \bar{w}_f, T_f, L_f, U_f, g_f, \pi_{\Phi f} , \pi_{\mathbb{Z} f})$ with a game signature $\mathcal{S}_f = (N_f, \mathcal{A}_f, X_f, \Phi_f )$ and $N_f = \{r_1, \cdot \cdot \cdot, r_\mathsf{k}\}$. Let us show how any finite GDLZ ST-model can be translated into a GDL ST-model.

\begin{definition}
Given the finite GDLZ ST-model $M_f$ and its signature $S_f$, 
we define the 
GDL ST-model $M_f' = (W_f, \bar{w}_f, T_f, $ $L_f', U_f', g_f, \pi_f')$ with a game signature $\mathcal{S}'_f = (N_f, \mathcal{A}'_f, \Phi'_f)$. The components $W_f, \bar{w}_f, T_f, g_f$ and $N_f$ are the same for $M_f$ and $M'_f$.

We construct $\Phi_f'$ over both $\Phi_f$, $X_f$ and its values. Although $X_f$ is a finite set, each one of its components has an integer value in each state $w \in W_f$.  As $\Phi_f'$ is finite, we construct it with the bounds $z_{min}$ and $ z_{max} \in \mathbb{Z}$. 
Since $\mathbb{Z}$ is a countable set, for any $z_{min}$ and $z_{max}$, we can define a finite enumeration of integer values.

The set of atomic propositions is defined as follows: 
$\Phi_f' = \{p, smaller(z_1, z_2), $ $ bigger(z_1, z_2),  equal(z_1, z_2), succ(z_1, z_2), prec(z_1, z_2), x(q): p \in \Phi, x \in X_f, z_{min} $ $ \leq q, z_1, z_2 \leq z_{max}\}$.

We define an action translation $Tr^a_f: \mathcal{A}_f \rightarrow \mathcal{A}_f'$ associating every action in $\mathcal{A}_f$ with an action in $\mathcal{A}'_f$ as follows:
\[Tr^a_f(a^r( z_1, \cdot \cdot \cdot, z_{l} )) = a^r_{z_1, \cdot \cdot \cdot, z_{l}}\]  
where $a^r( z_1, \cdot \cdot \cdot, z_{l}) \in \mathcal{A}_f, z_{min} \leq z_i \leq z_{max} $ and $ 0\leq i\leq l\}$.

Note that $Tr^a$ is an injective function.
Thereby, we can define the GDL components $\mathcal{A}_f'$ and $ L_f'$ based on $Tr^a$, as follows: (i) $\mathcal{A}_f' = \{Tr^a_f(a^r( z_1, \cdot \cdot \cdot, z_{l})): a^r( z_1, \cdot \cdot \cdot, z_{l} ) \in \mathcal{A}_f\}$; and (ii) $L_f' = \{(w, Tr^a_f(a)): (w, a) \in L_f\}$.

For each $w \in W_f$, each $r \in N_f$ and each joint action $( a^{r_1}, \cdot \cdot \cdot, a^{r_\mathsf{k}}) \in \prod_{r\in N_f} A^r_f$, where $A^r_f \in \mathcal{A}_f$, the update function is defined as:  
$
U_f'(w, ( Tr^a_f(a^{r_1}),$ $ \cdot \cdot \cdot, Tr^a_f(a^{r_\mathsf{k}}) )) = U_f(w, ( a^{r_1}, \cdot \cdot \cdot, a^{r_\mathsf{k}} ))
$. 

Let $\pi_{zf} \subset \Phi_f$ denote a set of propositions  describing the numerical order, such that $\pi_{zf} = \{succ(z, z+1), prec(z+1, z),  equal(z_1, z_1) : z_{min} \leq z < z_{max} \text{ \& } z_{min} \leq z_1 \leq z_{max}\} \cup \{smaller(z_1, z_2): z_{min}\leq z_1, z_2 \leq z_{max} \text{ \& } z_1<z_2\} \cup \{bigger(z_1, z_2): z_{min}\leq z_1, z_2 \leq z_{max} \text{ \& } z_1>z_2\}$.  

Finally, for all $w \in W_f$, we construct the valuation $\pi_f'$ as follows:
$\pi_f'(w) = \{ \pi_{\Phi f}(w) \cup \pi_{zf} \cup \{x(q): q\in\pi_{\mathbb{Z}f}(w),
x \in X_f
\}\}$.

We say that $M_f'$ is 
a bounded ST-model translation of $M_f$ and write $Tr^m_f(M_f)$.
\end{definition}

The path translation consists at assigning each action appearing on it to the appropriated GDL action through $Tr^a_f$.

\begin{definition}
\label{def:finitepathtr}
Define a path translation $Tr^\lambda_f: \mathcal{\delta}_f \rightarrow \mathcal{\delta}_f'$ associating every path $\delta_f = \bar{w}_f\stackrel{d_1}{\to} w_1 \stackrel{d_2}{\to}\cdot\cdot\cdot \stackrel{d_e}{\to} w_e$ in $M_f$ with a path $\delta_f'$ in $M_f'$:
$Tr^\lambda_f(\delta_f) = \bar{w}_f\stackrel{d_1'}{\to} w_1 \stackrel{d_2'}{\to}\cdot\cdot\cdot \stackrel{d_e'}{\to} w_e$, 
where  $d_i = ( a^{r_1}, \cdot \cdot \cdot, a^{r_\mathsf{k}} ) \in D_f$, $D_f = \prod_{r\in N_f} A^r_f$, $A^r_f \in \mathcal{A}_f$, $w_i \in W_f$, $d_i' = ( Tr^a_f(a^{r_1}), \cdot \cdot \cdot, Tr^a_f(a^{r_\mathsf{k}}) )$ and $1\leq i \leq e$.
\end{definition}

It follows that the translations of a finite GDLZ model and a path in a finite GDLZ model are a model and a path in GDL, respectively. 

\begin{proposition}
\label{prop:trfinite}
If $M_f$ is a finite GDLZ model then $Tr^m_f(M_f)$ is a GDL 
model.
\end{proposition}

\begin{proof}
Given $M_f$, $\mathcal{S}_f$, $Tr^m_f(M_f) = (W_f, \bar{w}_f, T_f, L_f',$ $ U_f', g_f, \pi_f')$, with $\mathcal{S}_f'= (N_f, $ $\mathcal{A}_f', \Phi_f')$, the integer bounds $z_{min}$ and $z_{max}$ and the construction of $Tr^m_f$, we have that  both the $W_f, \mathcal{A}_f'$ and $\Phi_f$ are ensured to be finite.
Since $Tr^a_f$ is an injective funcion, the proof proceeds in a similar way to the proof for Proposition \ref{prop:trmodel}.
\end{proof}

\begin{proposition}
\label{prop:finitepath}
If $\delta_f$ is a path in a finite GDLZ model $M_f$ then $Tr^\lambda_f(\delta_f)$ is a path in $Tr^m_f(M_f)$.
\end{proposition}

\begin{proof}
Given $Tr^a_f$ and $Tr^\lambda_f$, the proof proceeds as the proof for Proposition \ref{prop:trpath}.
\end{proof}

Next, we show a complete translate from GDLZ formulas to GDL formulas. Likewise to the
model translation, we use arbitrary bounds to restrict the numerical range in the formulas. 

\subsubsection{From bounded GDLZ Formulas to GDL Formulas. }
\label{sec:BoundedDescriptionsTranslation}


Assuming a GDLZ game signature $\mathcal{S}_f = (N_f, \mathcal{A}_f, \Phi_f, X_f)$, the semantics of a numerical variable $x \in X_f$  in a $\mathcal{L}_{GDLZ}$ formula is evaluated depending on the current game state. 

To translate the meaning of a numerical variable $x \in X_f$ occurring in an atomic formula $\varphi \in \mathcal{L}_{GDLZ}$ in the form $legal(a^r(\bar{z}))$, $does(a^r(\bar{z}))$, $\langle\bar{z}\rangle$, $z_1 < z_2$, $z_1 > z_2$ or $z_1 = z_2$, Algorithm \ref{alg:remove}, denoted $removeVar(\varphi)$, defines an intermediate formula $\varphi_x$ as the disjunction from all possible values $z_{min} \leq q \leq z_{max}$ for $x$ in $\varphi$ and $x(q)$. 
Algorithm \ref{alg:remove} stops when there is no more occurrence of numerical variables in the resulting formula.
%
 
\begin{algorithm}
\caption{$removeVar(\varphi)$}\label{alg:remove}
\hspace*{\algorithmicindent} \textbf{Input:} 
a  formula $\varphi \in  \mathcal{L}_{GDLZ}$. Assume the variable set $X_f$ and  $z_{min} \leq z_{max}$.
 \\
 \hspace*{\algorithmicindent} \textbf{Output:} a partially translated formula.
\begin{algorithmic}[1]

\State  $I \gets  \{z_{min}, \cdot \cdot \cdot, z_{max}\}$

\If {($\varphi = \text{``}legal(a^r(z_1, \cdot \cdot \cdot, z_m))\text{''}$)}
\ForEach {$z_i \in (z_1, \cdot \cdot \cdot, z_m)$} 
\If {$z_i \in  X_f$}
\Return 
$\bigvee_{q_i \in I} (removeVar(legal(a^r(z_1, \cdot \cdot \cdot, q_i, \cdot \cdot \cdot, z_m)) \land z_i(q_i))$
\EndIf
\EndFor

\ElsIf {($\varphi = \text{``}does(a^r(z_1, \cdot \cdot \cdot, z_m))\text{''}$)}
Proceeds as the previous case.

\ElsIf {($\varphi = \text{``}\langle z_1, \cdot \cdot \cdot, z_m\rangle\text{''}$)}
\ForEach {$z_i \in \langle z_1, \cdot \cdot \cdot, z_m\rangle $} 
\If {$z_i \in  X_f$}
\Return $\bigvee_{q_i \in I} (removeVar(\langle z_1, \cdot \cdot \cdot, q_i, \cdot \cdot \cdot, z_m\rangle) \land z_i(q_i))$
\EndIf
\EndFor

\ElsIf {($\varphi = \text{``}z_1 < z_2\text{''}$)}
\If {$z_1 \in  X_f$}
\Return $\bigvee_{q_1 \in I} (removeVar(q_1 < z_2) \land z_1(q_1))$
\EndIf

\If {$z_2 \in  X_f$}
\Return $ \bigvee_{q_2 \in I} (removeVar(z_1 < q_2) \land z_2(q_2))$
\EndIf


\ElsIf  {($\varphi = \text{``}z_1 > z_2\text{''} \text{ or } \varphi = \text{``}z_1 = z_2\text{''}$)}
Proceeds as the previous case.
\EndIf

\Return $\varphi$
\end{algorithmic}
\end{algorithm}

A numerical simple term $z_f$ is defined by $\mathcal{L}_{z_f}$, which is generated by the following BNF: 
\[z_f ::= z' \mid add(z_f, z_f)\mid sub(z_f,z_f) \mid min(z_f,z_f)\mid max(z_f,z_f)\] where $z' \in \mathbb{Z}$.  Note that $\mathcal{L}_{z_f} \subseteq \mathcal{L}_z$. Each numerical term $z_f \in \mathcal{L}_{z_f}$ occurring in a formula $\varphi \in \mathcal{L}_{GDLZ}$ is translated by its semantic interpretation through function $v_f$, defined in a similar way to Definition \ref{def:funcionv}:

\begin{definition}
\label{def:funcionvs}
Let us define function $v_f: \mathcal{L}_{z_f} \rightarrow \mathbb{Z}$, associating any $z_f \in \mathcal{L}_{z_f}$ to a number in $\mathbb{Z}$: 

\[v_f( z_f) = \begin{dcases}  z_i & \text{if }  z_f \in \mathbb{Z}
\\v_f(z_f') + v_f(z_f'') &\text{if } z_f = add(z_f', z_f'')
\\v_f(z_f') - v_f(z_f'') &\text{if } z_f = sub(z_f', z_f'')
\\minimum(v_f(z_f'),v_f(z_f'')) &\text{if } z_f = min(z_f', z_f'')
\\maximum(v_f(z_f'),v_f(z_f'')) &\text{if } z_f = min(z_f', z_f'') 
\end{dcases}\]
\end{definition}




The complete formula translation is restricted to bounded formulas, which are are defined as follows:

\begin{definition}
$\varphi \in \mathcal{L}_{GDLZ}$ is a bounded formula if, for any numerical term $z_f$ occurring in $\varphi$, we have $z_f \in \mathcal{L}_{z_f}$ and  $z_{min} \leq v_f(z) \leq z_{max}$  or if there is no occurrence of numerical terms in $\varphi$.
\end{definition}

We next define a translation map for
bounded formulas in $\mathcal{L}_{GDLZ}$ to formulas in $\mathcal{L}_{GDL}$. Each numerical simple term $z_f \in \mathcal{L}_{z_f}$ occurring in a formula
$\varphi \in \mathcal{L}_{GDLZ}$ is translated by its semantic interpretation through function $v_f$ (see
Definition \ref{def:funcionvs}).

\begin{definition}
\label{def:finiteformtr}
Given the GDLZ game signature $\mathcal{S}_f = (N_f, \mathcal{A}_f, X_f, \Phi_f )$ and function $v_f$, 
a translation $Tr^\varphi_f$ from a bounded formula $\varphi \in \mathcal{L}_{GDLZ}$  to a formula $\varphi' \in \mathcal{L}_{GDL}$ is defined as $Tr^\varphi_f = Tr^z_f(removeVar(\varphi))$, where $Tr^z_f$ is specified as follows:


\begin{itemize}
\item $Tr^z_f(\varphi) = \varphi$ for all $\varphi \in \Phi_f \cup \{initial, terminal, wins(r)\}$ $\cup $ $\{x(q) : x \in X_f, z_{min} \leq q \leq z_{max}\}$;
\item $Tr^z_f(\neg \varphi) = \neg Tr^z_f(removeVar(\varphi)))$;
\item $Tr^z_f(\varphi_1 \land \varphi_2) = Tr^z_f(removeVar(\varphi_1))) \land Tr^z_f(removeVar(\varphi_2)))$;
\item $Tr^z_f(\bigcirc \varphi) = \bigcirc Tr^z_f(removeVar(\varphi)))$;
\item $Tr^z_f(legal(a^r(\bar{z}))) = legal(Tr^a_f(a^r(v_f(z)  : z \in \bar{z} )))$;
\item $Tr^z_f(does(a^r(\bar{z}))) = does(Tr^a_f(a^r(v_f(z) : z \in \bar{z} )))$;
\item $Tr^z_f(\langle\bar{z}\rangle) = \bigwedge_{i = 1}^{|\bar{z}|} x_i(v_f(q_i))$;
\item $Tr^z_f(z_1 < z_2,) = smaller(v_f(z_1), v_f(z_2))$;
\item $Tr^z_f(z_1 > z_2) =  bigger(v_f(z_1), v_f(z_2))$;
\item $Tr^z_f(z_1 = z_2) = equal(v_f(z_1), v_f(z_2))$.
\end{itemize}

Where $r \in N_f$, $x_i \in X_f, q_i$ is the $i$-th value in $ \bar{z}$ and $ 0\leq i \leq |\bar{z}|$.
\end{definition}

Let us illustrate the translation of GDLZ formulas into GDL using $Tr^\varphi_f$.
\begin{example}

Let $I = \{z_{min}, \cdot \cdot \cdot, z_{max}\}$ and $\varphi_1 = does(reduce^r(heap_1, add(1,2)))$, where $heap_1 \in X_f$, then $Tr^\varphi_f(\varphi_1) = \bigvee_{h_1 \in  \{z_{min}, \cdot \cdot \cdot, z_{max}\}} (does(reduce^r(h_1, 3)) \land heap_1(h_1) )$.
 

 
\end{example}

The translation of a GDLZ formula is a GDL formula. Furthermore, if the GDLZ formula is valid at a stage in the path in a finite GDLZ model, then its translation will be valid at the same stage in the translated path in the translated model.


\begin{proposition}
\label{prop:finiteformula}
If $\varphi$ $\in \mathcal{L}_{GDLZ}$  then $Tr^\varphi_f(\varphi)$ $\in \mathcal{L}_{GDL}$.
\end{proposition}

\begin{proof}
Given a finite GDLZ model $M_f$ with the game sinature $\mathcal{S} = (N, \mathcal{A}, \Phi, X)$ and $M_f' = Tr^m_f(M_f)$, with $\mathcal{S}_f' = (N_f, \mathcal{A}_f', \Phi_f')$.
Assume that $\varphi \in \mathcal{L}_{GDLZ}$, since $Tr^\varphi_f(\varphi) = Tr^z_f (removeVar(\varphi))$, we need to show that  $Tr^z_f(removeVar(\varphi)) \in \mathcal{L}_{GDL}$ for each form of $\varphi$.
If there is a numerical variable $x$ in an atomic formula $\varphi$, the method $removeVar(\varphi)$ constructs $\varphi_x$ as a disjunction from $\varphi$ with every possible value of $x$ between $z_{min}$ and $z_{max}$  and the proposition $x(q)$.
By $Tr^m_f$ definition, $\{x(q): x \in X, z_{min} \leq q \leq z_{max}\} \subseteq \Phi' $ and thereby $x(q) \in \mathcal{L}_{GDL}$. The translation $Tr^\varphi_f$ proceeds assigning each subformula of $removeVar(\varphi)$ to a $\mathcal{L}_{GDL}$ formula.
The proof proceeds as the proof for Proposition \ref{prop:trformula}.
\end{proof}

\begin{theorem}
\label{prop:finitevalid}
If  $M_f$ is a finite GDLZ ST-model,  $\varphi \in \mathcal{L}_{GDLZ}$ is a bounded formula and $M_f, \delta_f, j \models\varphi$  then $Tr^m_f(M_f), $ $Tr^\lambda_f(\delta_f), j \models Tr^\varphi_f(\varphi)$. 
\end{theorem}

\begin{proof}
Given a finite GDLZ ST-model $M_f$ $= (W_f, \bar{w}_f, T_f, L_f, U_f, g_f, \pi_{\Phi f}, \pi_{\mathbb{Z} f})$,
with the game signature $\mathcal{S} = (N, \mathcal{A},$ $ X, \Phi )$, 
a complete path $\delta_f$, a stage $j$ on $\delta_f$, a formula $\varphi \in \mathcal{L}_{GDLZ}$ and the function $v_f$. Let $Tr^m_f(M_f) = M_f' = (W_f, \bar{w}_f, T_f, L_f',$ $ U_f', g_f, \pi_f')$, 
$\delta_f' = Tr^\lambda_f(\delta)$, $\varphi' = Tr^\varphi_f(\varphi)$ and  $z_{min} \leq z_{max} \in \mathbb{Z}$ denote the integer bounds in $Tr^m$. 


The proof is performed in a similar way that in the proof for Theorem \ref{prop:valid}, except in the case where there are numerical variables occurring in $\varphi$. Lets consider the case where $\varphi$ is in the form $legal(a^r(z_1, \cdot\cdot \cdot, z_m))$ and we have only one numerical variable $z_i \in X_f$ occurring in the parameter list $(z_1, \cdot \cdot \cdot, z_m)$,  where $1 \leq i \leq m$. 
By  $Tr^\varphi_f$ and Algorithm \ref{alg:remove} definition,  $Tr^\varphi_f(\varphi) = Tr^z_f(removeVar(\varphi)) = \bigvee_{q_i \in \{z_{min}, \cdot \cdot \cdot, z_{max}\}} (legal(Tr^a_f(a^r(z_1, \cdot\cdot \cdot, q_i\cdot\cdot \cdot,  z_m))) \land z_i(q_i))$. 
For any $w \in W_f'$, $z_i \in X_f$ and $z_{min} \leq q_i' \leq z_{max} $, we have that $z_i(q_i') \in \pi_f'(w)$ iff $q_i'$ is the $i$-th value of $\pi_{\mathbb{Z}f}(w)$, i.e., variable $z_i$ has the value $q_i'$ in state $w$. 
Thereby, $(legal(Tr^a_f(a^r(z_1, \cdot\cdot \cdot, q_i\cdot\cdot \cdot,  z_m))) \land z_i(q_i))$ will hold only in the case where $q_i = q_i'$. Thus, $Tr^m_f(M_f), Tr^\lambda_f(\delta_f), j\models \bigvee_{q_i \in \{z_{min}, \cdot \cdot \cdot, z_{max}\}} (legal(Tr^a_f(a^r(z_1, \cdot\cdot \cdot, q_i\cdot\cdot \cdot,  z_m))) \land z_i(q_i))$ iff $M_f, \delta_f, j \models legal(a^r(z_1, \cdot \cdot \cdot, z_m))$.
Since $removeVar(legal(a^r(z_1, \cdot \cdot \cdot, z_m)))$ will be recursively applied to every $z_i \in X_f$ occurring in $(z_1, \cdot \cdot \cdot, z_m)$, it is easy to see that the result holds when we have two or more numerical variables in $legal(a^r(z_1, \cdot \cdot \cdot, z_m))$. The proof proceeds in a similar way if $\varphi$ is either in the form $does(a^r(z_1, \cdot \cdot \cdot, z_m))$, $\langle z_1, \cdot \cdot \cdot, z_m\rangle$, $z_1 < z_2$, $z_1 > z_2$ or $z_1 = z_2$.
\end{proof}

In the next section, we briefly describe how to translate GDL ST-models into GDLZ ST-models. Besides that, we show that GDL is a sublanguage of GDLZ.

\subsection{From GDL to GDLZ}
\label{sec:GDLTranslation}

Conversely, we show that any GDL ST-model can be transformed into a GDLZ ST-model. Given a GDL ST-model $M' = (W, \bar{w}, T,L, U, g, \pi')$ with a game signature $\mathcal{S}' = (N, \mathcal{A}, \Phi)$, we define an associated GDLZ ST-model $M = (W, \bar{w}, $\linebreak$ T, L, U, g, \pi_\Phi , \pi_{\mathbb{Z}})$ with the game signature $\mathcal{S} = (N, \mathcal{A}, X, \Phi )$, such that all elements are the same, except by $\pi_\Phi, \pi_{\mathbb{Z}}$ and X 
and $X $. 
These GDLZ components are defined as follows:
(i) $\pi_\Phi(w) = \pi'(w)$; (ii) $\pi_{\mathbb{Z}}(w) = \emptyset$; and (iii) $X = \emptyset$.



It follows that any formula $\varphi \in \mathcal{L}_{GDL}$ is also a formula in GDLZ, i.e. $\varphi \in \mathcal{L}_{GDLZ}$.

\begin{proposition}
\label{prop:contain}
If $\mathcal{S}' = (N, \mathcal{A}, \Phi')$ and $\mathcal{S} = (N, \mathcal{A}, X, \Phi )$ are GDL and GDLZ game signatures, respectively, 
and $\Phi' \subseteq \Phi$, then $\mathcal{L}_{GDL} \subseteq \mathcal{L}_{GDLZ}$.
\end{proposition}

\begin{proof}
Assume the GDL and GDLZ 
signatures  $\mathcal{S}' = (N, \mathcal{A}, \Phi')$ and $\mathcal{S} = (N, \mathcal{A}, $ $X, \Phi )$, respectively, and $\Phi' \subseteq \Phi$, we show that for any $\varphi \in \mathcal{L}_{GDL}$, $\varphi \in \mathcal{L}_{GDLZ}$.

Assume $\varphi \in \mathcal{L}_{GDL}$, if $\varphi $ is of the form $p, initial, terminal , wins(r),  \neg \varphi, \varphi \land \varphi$ or $ 
\bigcirc \varphi$, where $p \in \Phi'$ and $r \in N$, by the grammar definition of GDLZ, since $\Phi' \subseteq \Phi$, we can easily see that  $\varphi \in \mathcal{L}_{GDLZ}$.
Otherwise, if $\varphi $ is of the form 
 $legal(a ^r)$ or $does(a ^r)$, where $a^r \in \mathcal{A}, r \in N$, 
we have that $legal(a ^r(\bar{z}))$, $does(a ^r(\bar{z})) \in \mathcal{L}_{GDLZ}$. By the numerical list $\bar{z}$ grammar, we know that $\bar{z}$ can be empty. Therefore, $legal(a ^r(\varepsilon))$,  $does(a ^r(\varepsilon)) \in \mathcal{L}_{GDLZ}$ or simply $legal(a ^r)$,  $does(a ^r) \in \mathcal{L}_{GDLZ}$. 
Thus, $\mathcal{L}_{GDL} \subseteq \mathcal{L}_{GDLZ}$. 
\end{proof}

\subsection{Succinctness}
\label{sec:succincteness}

Next, we compare $\mathcal{L}_{GDLZ}$ and $\mathcal{L}_{GDL}$ in order to show the succinctness of  $\mathcal{L}_{GDLZ}$ in describing the same game. 
The following definition specifies when two sets of formulas in GDLZ and GDL describe the same game.

\begin{definition}
\label{def:samegame}
Two sets of formulas $\Sigma_{GDLZ}\subseteq\mathcal{L}_{GDLZ}$ and $\Sigma_{GDL} \subseteq \mathcal{L}_{GDL}$ describe the same game 
either (i) if  $\Sigma_{GDLZ} = \{\varphi: \varphi \in \Sigma_{GDL}\}$ and  $\mathcal{L}_{GDL}$ and $\mathcal{L}_{GDL}$ have, respectively, 
$\mathcal{S}' = (N, \mathcal{A}, \Phi)$ and $\mathcal{S} = (N, \mathcal{A}, \emptyset, \Phi )$; (ii) if $\Sigma_{GDL} = \{Tr^\varphi(\varphi, \delta[j]): \varphi \in \Sigma_{GDLZ}\}$, given a GDLZ ST-model $M$, a path $\delta$ in $M$ and a stage $j$ in $\delta$ or (iii) if $\Sigma_{GDL} = \{Tr^\varphi_f(\varphi): \varphi \in \Sigma_{GDLZ} \}$, where every $\varphi \in  \Sigma_{GDLZ} $ is a bounded formula.  
\end{definition}

The following theorem show that (i) a GDLZ description has less subformulas and (ii) if we compare with the path translation, the growth is linear, if we compare with the complete translation, the growth is exponential. 

\begin{theorem}
\label{prop:Compact}
If $\Sigma_{GDLZ}$ and $\Sigma_{GDL}$ are two sets of formulas in $\mathcal{L}_{GDLZ}$ and $\mathcal{L}_{GDL}$, respec., describing the same game, then $|Sub(\Sigma_{GDLZ})| \leq |Sub(\Sigma_{GDL})|$. 
\end{theorem}

\begin{proof}
Assume the GDL and GDLZ game signatures  $\mathcal{S}' = (N, \mathcal{A}', \Phi')$ and $\mathcal{S} = (N, \mathcal{A}, X, \Phi )$, respectively. 
Since $\Sigma_{GDLZ}$ and $\Sigma_{GDL}$ describe the same game, by Definition \ref{def:samegame}, we have either: (i) $\Sigma_{GDLZ} = \{\varphi: \varphi \in \Sigma_{GDL}\}$,
$\mathcal{S}' = (N, \mathcal{A}, \Phi)$ and $\mathcal{S} = (N, \mathcal{A}, \emptyset, \Phi )$;  (ii) $\Sigma_{GDL} = \{Tr^\varphi(\varphi, \delta[j]): \varphi \in \Sigma_{GDLZ}\}$, for a GDLZ ST-model $M$, a path $\delta$ in $M$ and a stage $j$ in $\delta$, or (iii) if $\Sigma_{GDL} = \{Tr^\varphi_f(\varphi): \varphi \in \Sigma_{GDLZ} \}$, where every $\varphi \in  \Sigma_{GDLZ} $ is a bounded formula.  
In the first case, $\mathcal{A}' = \mathcal{A}$, $\Phi' = \Phi$, $X = \emptyset$ and
$\Sigma_{GDLZ} = \{\varphi: \varphi \in \Sigma_{GDL}\}$, we clearly have $|\Sigma_{GDLZ}| = |\Sigma_{GDL}|$ and $|Sub(\Sigma_{GDLZ})| = |Sub(\Sigma_{GDL})|$.

Given a path $\delta$ in a GDLZ ST-model $M$ and a stage $j$, let us now consider the case (ii) where $\Sigma_{GDL} = \{Tr^\varphi(\varphi, \delta[j]): \varphi \in \Sigma_{GDLZ}\}$.  
From $Tr^\varphi(\varphi, \delta[j])$, we have that any translation assigns $\varphi$ to a corresponding $\varphi'$ 
where $|Sub(\varphi)| = |Sub(\varphi')|$, except in the case where $\varphi$ is of the form $\langle \bar{z} \rangle$.
If $\varphi$ is of the form $\langle\bar{z}\rangle $, 
then $\varphi'$ will be constructed as $\bigwedge_{i = 1}^{|\bar{z}|} x_i(v(q_i, w))$, 
where $x_i \in X$ and $q_i$ is the $i$-th value of $\bar{z}$. Thus, $|Sub(\varphi')| = |\bar{z}||Sub(\varphi)|$. Since $|Sub(\varphi)| = 1$, then  $|Sub(\varphi')| = |\bar{z}|$.

Denote $\Sigma_l = \Sigma_{GDLZ} - \{\langle \bar{z}\rangle : \langle \bar{z}\rangle \in \mathcal{L}_{GDLZ}\}$, i.e. $\Sigma_l$ is the subset of $\Sigma_{GDLZ}$ without any formula $\langle \bar{z} \rangle$. Thereby $|Sub(\Sigma_l)| = |Sub(\{ Tr^\varphi(\varphi, \delta[j]): \varphi \in \Sigma_l)\}|)$.
Assuming $k$ as  the amount of formulas in the form $\langle \bar{z} \rangle \in \Sigma_{GDLZ}$, we have 
$|Sub(\Sigma_{GDL})| = |Sub(\Sigma_l)| + |\bar{z}|k  $.
Thereby, in the second case,
we have $|Sub(\Sigma_{GDLZ})| \leq |Sub(\Sigma_{GDL})|$.

Let us consider case (iii), where $\Sigma_{GDL} = \{Tr^\varphi_f(\varphi): \varphi \in \Sigma_{GDLZ} \}$ and every $\varphi \in  \Sigma_{GDLZ} $ is a bounded formula. Let $\mu = z_{max}-z_{min}$. The proof for case (iii) proceeds in the same way that for case (ii), except in the situation where there are numerical variables occurring in any $\varphi \in  \Sigma_{GDLZ} $. 
If we have at least one numerical variable occurring in $\varphi$, we know that $\varphi$ is either in the form $legal(a^r(\bar{z}))$, $does(a^r(\bar{z}))$, $\langle \bar{z}\rangle$, $z_1<z_2$, $z_1>z_2$ or $z_1 = z_2$. Thereby, $|\varphi| = 1$ and
$|Tr^\varphi_f(removeVar(\varphi))| = 2\mu^{\eta}\times|\varphi|   $, where $\eta$ 
is the amount of numerical variables occurring in $\varphi$.
Thereby, 
$|\varphi| <  |Tr^\varphi_f(removeVar(\varphi))|$ 
and
$|Sub(\Sigma_{GDLZ})| \leq |Sub(\Sigma_{GDL})|$.

Denote 
$\Sigma_l' = \Sigma_{GDLZ} - \{\langle \bar{z}\rangle : \langle \bar{z}\rangle \in \mathcal{L}_{GDLZ}\} - \{\varphi \in \mathcal{L}_{GDLZ}: $ there is at least one numerical variable in $\varphi\}$.  
Assuming $k$ as the amount of formulas in the form $\langle \bar{z} \rangle \in \Sigma_{GDLZ}$ and $\kappa$ as the amount of formulas where occurs $\eta$ numerical variables, we have 
$|Sub(\Sigma_{GDL})| = |Sub(\Sigma_l')| +2\mu^{\eta}\kappa +  |\bar{z}|k $.
\end{proof}

\begin{theorem}
\label{prop:linear}
Given $\Sigma_{GDLZ} \subseteq \mathcal{L}_{GDLZ}$, a GDLZ ST-model $M$ with the game signature $\mathcal{S} = (N, \mathcal{A}, \Phi, X)$: 

\begin{enumerate}

\item  If $\Sigma_{GDL} = \{Tr^\varphi(\varphi, \delta[j]): \varphi \in \Sigma_{GDLZ}\}$, given a path $\delta$ in $M$ and a stage $j$ in $\delta$,  then $|Sub(\Sigma_{GDL})|$ grows in the order $\mathcal{O}(n)$, where
$n = |Sub(\Sigma_{l})| +|X|k$,  the value $k$ represents the amount of formulas in the form $\langle \bar{z} \rangle $ in $ \Sigma_{GDLZ}$ 
and  $\Sigma_l = \Sigma_{GDLZ} - \{\langle \bar{z}\rangle : \langle \bar{z}\rangle \in \mathcal{L}_{GDLZ}\}$, i.e. $\Sigma_l$ is the subset of $\Sigma_{GDLZ}$ without any formula $\langle \bar{z} \rangle$;

\item If $\Sigma_{GDL} = \{Tr^\varphi_f(\varphi): \varphi \in \Sigma_{GDLZ} \}$, where every $\varphi \in  \Sigma_{GDLZ} $ is a bounded formula, then  $|Sub(\Sigma_{GDL})|$ grows in the order $\mathcal{O}(n + \kappa\mu^\eta)$, where $n = |Sub(\Sigma_{l})| +|X|k$,  the value $k$ represents the amount of formulas in the form $\langle \bar{z} \rangle $ in $ \Sigma_{GDLZ}$,  $\Sigma_l = \Sigma_{GDLZ} - \{\langle \bar{z}\rangle : \langle \bar{z}\rangle \in \mathcal{L}_{GDLZ}\}$ and $\eta$ is the amount of numerical variables occurring in $\kappa$ variables. 
\end{enumerate}

\end{theorem}

\begin{proof}
Theorem \ref{prop:linear} follows directly from the proof of Theorem \ref{prop:Compact}.
\end{proof}

The partial translation $Tr^\varphi$ only concerns a fragment of the GDLZ model, that is the part of the model involved in a specific path. The size of a formula translated through $Tr^\varphi$ has a linear growth over the number of numerical variables in $X$ and the number of formulas in the form $\langle \bar{z} \rangle$. Conversely, $Tr^\varphi_f$ is a complete translation over finite GDLZ models. To represent a GDLZ formula in a GDL formula regardless of a specific path, we should remove the occurrence of numerical variables as numerical terms (see Algorithm \ref{alg:remove}). This procedure exponentially increases the size of the translated formula, depending mainly on the occurrence of numerical variables in the original GDLZ formula.
 
\section{Conclusion}
\label{sec:conclusions}

In this paper, we have introduced a GDL extension to describe games with numerical aspects, called GDLZ. 
In GDLZ, states are evaluated with propositions and an assignment of integer values to numerical variables.
This allows us to define the terminal and goal states in terms of the numerical conditions. 
Furthermore, we define actions with numerical parameters, such that these parameters can influence over the action legality and over the state update. 
The language was extended mainly to include the representation of numerical variables and integer values as well as to allow numerical comparison.



We defined translations between GDLZ and GDL game models and formulas. 
Since GDL models have finite components, we can not define a complete model translation for any GDLZ model.
We first defined a partial translation from any GDLZ model restricted to a specified path, i.e. only a run in the game is represented. Second, we  defined a complete translation from GDLZ models with finite components and bounded formulas. 
We show that, in both cases, a translated GDLZ model, path or formula is a GDL model, path or formula, respectively. Furthermore, we prove that if a formula is satisfied at a stage in a path under a GDLZ model, its translation will also be satisfied at the same stage in the translated path under the translated model. 

Finally, we show that, if we have a GDLZ and a GDL description for the same (finite) game, the GDLZ description is more succinct or equal, in terms of the quantity of subformulas in the description. More precisely,
if the GDL game description is based on the partial translation from a GDLZ description restricted to one path, it is linearly larger then the GDLZ description. When we consider  the complete model translation, 
the GDL description is exponentially larger than the GDLZ description. 

Future work may extend GDLZ to define numerical rewards to players, stating their achievement when the game ends. It means that numerical variables may not have values assigned in some state of the model. Our aim is to investigate this new kind of numerical models.
In our framework, it is possible to define both concurrent and sequential games. However, the legality of an agent's action is independent from the actions of other agents. Thereby, it may be inappropriate to describe concurrent games where the actions of two agents change the same numerical variable. To overcome this limitation, future work may explore the definition of the legality function over joint actions.


%

\subsubsection{Acknowledgments.} 
Munyque Mitttelmann and Laurent Perrussel acknowledge the support of the ANR project AGAPE ANR-18-CE23-0013.

\label{sec:bib}

\bibliographystyle{splncs04}

\begin{thebibliography}{10}
\providecommand{\url}[1]{\texttt{#1}}
\providecommand{\urlprefix}{URL }
\providecommand{\doi}[1]{https://doi.org/#1}

\bibitem{Genesereth2005}
Genesereth, M., Love, N., Pell, B.: {General game playing: Overview of the AAAI
  competition}. AI magazine  \textbf{26}(1),  1--16 (2005),
  \url{http://www.aaai.org/ojs/index.php/aimagazine/article/viewArticle/1813}

\bibitem{Gerevini2008}
Gerevini, A.E., Saetti, A., Serina, I.: {An approach to efficient planning with
  numerical fluents and multi-criteria plan quality}. Artificial Intelligence
  \textbf{172}(8-9),  899--944 (2008). \doi{10.1016/j.artint.2008.01.002}

\bibitem{Ghallab98}
Ghallab, M., Howe, A., Knoblock, C., Mcdermott, D., Ram, A., Veloso, M., Weld,
  D., Wilkins, D.: {PDDL - The Planning Domain Definition Language}. Tech.
  rep., AIPS-98 Planning Competition Committee (1998),
  \url{http://citeseerx.ist.psu.edu/viewdoc/summary?doi=10.1.1.37.212}

\bibitem{Giacomo2010}
Giacomo, G.D., Lesp, Y., Pearce, A.R.: {Situation Calculus-Based Programs for
  Representing and Reasoning about Game Structures}. In: Proc. of the Twelfth
  International Conference on the Principles of Knowledge Representation and
  Reasoning (KR 2010). pp. 445--455 (2010)

\bibitem{Jiang2014GDLmeetsATL}
Jiang, G., Zhang, D., Perrussel, L.: {GDL Meets ATL: A Logic for Game
  Description and Strategic Reasoning}. In: Pham, D.N., Park, S.B. (eds.)
  PRICAI 2014: Trends in Artificial Intelligence. pp. 733--746. Springer Int.
  Publishing, Cham (2014)

\bibitem{Jiang2016a}
Jiang, G., Zhang, D., Perrussel, L., Zhang, H.: {Epistemic GDL: A logic for
  representing and reasoning about imperfect information games}. IJCAI
  International Joint Conference on Artificial Intelligence
  \textbf{2016-Janua},  1138--1144 (2016)

\bibitem{gdl_specification}
Love, N., Genesereth, M., Hinrichs, T.: {General Game Playing: Game Description
  Language Specification}. Tech. Rep. LG-2006-01, Stanford University,
  Stanford, CA (2006), \url{http://logic.stanford.edu/reports/LG-2006-01.pdf}

\bibitem{Fox2003}
{Maria Fox}, {Derek Long}: {PDDL2.1: An extension to PDDL for expressing
  temporal planning domains}. Journal of Artificial Intelligence Research
  \textbf{20},  1--48 (2003) ,
  \url{http://citeseerx.ist.psu.edu/viewdoc/summary?doi=10.1.1.68.1957}

\bibitem{McDermott2000}
McDermott, D.M.: {The 1998 AI Planning Systems Competition}. AI Magazine
  \textbf{21}(2), ~35 (2000). \doi{10.1609/AIMAG.V21I2.1506}

\bibitem{Parikh1985}
Parikh, R.: {The Logic of Games and its Applications}. North-Holland
  Mathematics Studies  \textbf{102}(C),  111--139 (1985).
  \doi{10.1016/S0304-0208(08)73078-0}

\bibitem{Pauly2003}
Pauly, M., Parikh, R.: {Game Logic - An Overview}. Studia Logica
  \textbf{75}(2),  165--182 (nov 2003). \doi{10.1023/A:1027354826364}

\bibitem{Schiffel2014}
Schiffel, S., Thielscher, M.: {Representing and reasoning about the rules of
  general games with imperfect information}. Journal of Artificial Intelligence
  Research  \textbf{49},  171--206 (2014)

\bibitem{Thielscher2010}
Thielscher, M.: {A general game description language for incomplete information
  games}. Proceedings of the Twenty-Fourth AAAI Conference on Artificial
  Intelligence (AAAI-10) pp. 994--999 (2010),
  \url{https://www.aaai.org/ocs/index.php/AAAI/AAAI10/paper/view/1727}

\bibitem{Thielscher2016}
Thielscher, M.: {GDL-III: A proposal to extend the game description language to
  general epistemic games}. In: Proceedings of the European Conference on
  Artificial Intelligence (ECAI). vol.~285, pp. 1630--1631. Hague (2016).
  \doi{10.3233/978-1-61499-672-9-1630}

\bibitem{Thielscher2017}
Thielscher, M.: {GDL-III: A description language for epistemic general game
  playing}. IJCAI Int. Joint Conference on Artificial Intelligence pp.
  1276--1282 (2017)

\bibitem{VanBenthem2001}
{Van Benthem}, J.: {Games in dynamic-epistemic logic}. Bulletin of Economic
  Research  \textbf{53}(4),  219--248 (2001). \doi{10.1111/1467-8586.00133}

\bibitem{VanBenthem2008}
{Van Benthem}, J., Ghosh, S., Liu, F.: {Modelling simult. games in dynamic
  logic}. Synthese  \textbf{165}(2),  247--268 (2008).
  \doi{10.1007/s11229-008-9390-y}

\bibitem{Zhang2015Representing}
Zhang, D., Thielscher, M.: {Representing and Reasoning about Game Strategies}.
  Journal of Philosophical Logic  \textbf{44}(2),  203--236 (2014).
  \doi{10.1007/s10992-014-9334-6}

\end{thebibliography}

\end{document}